\documentclass[10pt,conference]{IEEEtran}
\IEEEoverridecommandlockouts

\usepackage{caption}
\usepackage{subcaption}
\usepackage{tcolorbox}
\usepackage{multirow}
\usepackage{xcolor}
\usepackage{colortbl}
\usepackage{comment}
\usepackage{booktabs}
\usepackage[ruled,linesnumbered, vlined]{algorithm2e}
\usepackage{setspace}
\newcommand{\mytcc}[1]{\tcc*[r]{\textcolor{brown}{#1}}}

\usepackage{enumitem}
\usepackage{multicol}
\usepackage{parcolumns}
\usepackage{threeparttable}
\usepackage{makecell}
\usepackage{textcomp}
\usepackage{stfloats}
\usepackage{pifont}

\usepackage{amsmath,amssymb,amsfonts}
\usepackage{amsthm}

\newtheorem{theorem}{Theorem}

\newtheorem{definition}{Definition}
\newtheorem{example}{Example}

\usepackage{graphicx}
\usepackage{textcomp}
\usepackage{mathrsfs}
\usepackage{bm}
\usepackage[
  colorlinks=true,
  linkcolor=blue,
  citecolor=blue,
  urlcolor=blue
]{hyperref}
\usepackage{xspace}
\usepackage{url}

\usepackage{xcolor}
\newcommand{\eqdef}{\buildrel \mbox{\tiny\textrm{def}} \over =}
\newcommand{\mrm}{\mathrm}

\newcommand{\name}{\textsc{ProF}\xspace}

\def\BibTeX{{\rm B\kern-.05em{\sc i\kern-.025em b}\kern-.08em
    T\kern-.1667em\lower.7ex\hbox{E}\kern-.125emX}}
\begin{document}

\title{Provable Fairness Repair for Deep Neural Networks}

\IEEEoverridecommandlockouts

\author{
\IEEEauthorblockN{Jianan Ma}
\IEEEauthorblockA{\textit{Hangzhou Dianzi University, China} \\
\textit{Zhejiang University, China} \\
majianannn@gmail.com}
\and
\IEEEauthorblockN{Jingyi Wang\IEEEauthorrefmark{4}}
\IEEEauthorblockA{\textit{Zhejiang University, China} \\
wangjyee@zju.edu.cn}
\and[\hfill\mbox{}\vspace{1em}\\
\mbox{}\hfill]
\IEEEauthorblockN{Qi Xuan}
\IEEEauthorblockA{\textit{Zhejiang University of Technology, China} \\
xuanqi@zjut.edu.cn}
\and[\hfill]
\IEEEauthorblockN{Zhen Wang}
\IEEEauthorblockA{\textit{Hangzhou Dianzi University, China} \\
wangzhen@hdu.edu.cn}
\thanks{\IEEEauthorrefmark{4}Corresponding author: Jingyi Wang.}
}

\maketitle

\begin{abstract} 
Deep neural networks (DNNs) are suffering from ethical issues such as individual discrimination.
In response, extensive NN repair techniques have been developed to adjust models and mitigate such undesired behaviors.
However, existing fairness repair methods are typically data-centric,
which often lack provable guarantees and generalization to unseen samples.
To overcome these limitations, we propose \name, a novel fairness repair framework with provable guarantees. The key intuition of \name is to leverage interval bound propagation (a widely used NN verification technique) to soundly capture model outputs over the whole set $\mathcal{S}(\bm{x})$ around a biased sample $\bm{x}$.
The derived bounds are utilized to guide fairness repair which encourages the model to produce consistent outputs on $\mathcal{S}(\bm{x})$.
Specifically, we integrate fairness constraints and model modifications into a unified constraint-solving formulation, which can be transformed to a Mixed-Integer Linear Programming (MILP) problem solvable by off-the-shelf solvers. 
The solution to the MILP problem effectively induces a repaired model with guaranteed fairness over the whole set $\mathcal{S}(\bm{x})$.
We evaluate \name on four widely used benchmark datasets and demonstrate that it achieves provable fairness repair, with generalization of up to 95.93\% on full datasets and 93.16\% on the entire input space. 
Notably, \name can be easily configured to support multiple sensitive attributes and more practical fairness definitions, while providing provable repair guarantees and delivering around 90\% fairness improvement.
Our code is available at \url{https://github.com/nninjn/ProF}.

\end{abstract}


\section{Introduction}
Deep neural networks (DNNs) have demonstrated impressive performance in a wide range of applications such as computer vision~\cite{badar2020application}, natural language processing~\cite{devlin2018bert}, and autonomous driving~\cite{deng2021deep}. 
Despite the remarkable success, their increased adoption in sensitive domains (e.g., crime risk assessment~\cite{brennan2009evaluating}, public policy~\cite{rodolfa2021empirical}, and credit scoring~\cite{khandani2010consumer}) has raised concerns that DNNs learning from biased data can produce unfair results.
Individual discrimination has become one of the most critical problems, where input pairs differing only in sensitive attributes (e.g., gender, age, etc.) receive different predictions~\cite{aggarwal2019black}.
In response to this threat, numerous studies on DNN fairness testing~\cite{zhang2020white, quan2025dissecting, wang2024maft, zhang2021efficient, zheng2022neuronfair} and verification~\cite{biswas2023fairify, kim2024fairquant, sun2021probabilistic} have been proposed.
These techniques aim to either uncover unfairness by generating discriminatory instances or to provide guarantees through formal methods.
However, their analysis results cannot provide further solutions as they do not directly mitigate the unfairness.
This raises a crucial question: \emph{how to repair the model (ultimately with provable guarantee) once biased behavior is identified?}

A typical method for DNN fairness repair is to retrain the model with a set of discriminatory inputs. 
While retraining is easy to deploy, it inherently faces challenges in real-world applications due to its high costs and the necessity of accessing the original training set, which may be impractical when the model is obtained from a third party or the training data are private.
For more effective and efficient repair, researchers have proposed various methods that aim to adjust the parameters of a biased NN to eliminate discriminatory behaviors.
Among them, a common paradigm inspired by traditional software programs debugging is to first identify the key units in the model that are responsible for discrimination.
For example, CARE~\cite{sun2022causality} constructs the causal model between neurons and model outputs to pinpoint problematic neurons, followed by using a heuristic algorithm to generate neuron-level patches to modify their parameters.
In addition, some other approaches achieve repair through more efficient ways: 
RUNNER~\cite{li2024runner} design a loss function to iteratively optimize the identified biased neurons via gradient descent, while IDNN~\cite{chen2024isolation} directly isolate them.
NeuFair~\cite{dasu2024neufair}, on the other hand, leverages the simulated annealing algorithm to provide statistical guarantees for group fairness repair. 
More recently, GRFT~\cite{quan2025dissecting} has been proposed for more effective fairness testing and repair. 
It achieves state-of-the-art performance by directly minimizing the distance of model outputs between discriminatory inputs and their similar instances.

\textbf{Research gap} - 
While prior repair methods have shown effectiveness in certain cases, they suffer from two fundamental limitations.
First, the repair techniques they employ (e.g., heuristic algorithms, gradient-based parameter tuning, or neuron isolation) are empirical in nature and lack deterministic guarantees. 
Even on the discriminatory pair $\langle \bm{x}, \bm{x}' \rangle$ used in repair, these techniques may not ensure that the model’s outputs satisfy fairness constraints.
Second, these approaches use a set of discriminatory input pairs $\langle \bm{x}, \bm{x}^{\prime} \rangle$ to
analyze and repair model behavior, where $\bm{x}^{\prime}$ is chosen from $\mathcal{S}(\bm{x})$ (the set of all instances similar to $\bm{x}$). 
As such, they only observe and correct unfairness over a limited region of $\mathcal{S}(\bm{x})$.
Consequently, the repaired model may still exhibit unfair behavior on unseen or unsampled inputs, resulting in limited generalization.

\textbf{Our insight} - To address the above challenges, our key idea is to leverage interval bound propagation (a widely used NN verification technique) to soundly capture model outputs over $\mathcal{S}(\bm{x})$.
The derived bounds can be utilized to guide fairness repair, which encourages the model to produce consistent outputs across all similar instances.
To further provide guarantees for repair, our intuition is to exploit these bounds as a bridge to integrate fairness constraints and model modifications to construct a unified constraint-solving formulation. 
The solution to this problem can induce a repaired model with guaranteed fairness over $\mathcal{S}(\bm{x})$.

\textbf{Our solution - }
Based on the above insight, we propose \name, a novel provable NN fairness repair framework.
As illustrated in Fig.~\ref{fig:overview}, \name consists of two core components.
Given a DNN $f$ (which can be sliced as the first $\mathrm{L}-1$ layers $f_{1:\mathrm{L}}$ and the last layer $f_{\mathrm{L}}$), the first step applies interval bound propagation to calculate the concrete bounds that soundly capture the outputs of $f_{1:\mathrm{L}}$ (i.e., outputs in feature space).
These bounds define axis-aligned hyperrectangles, which enables us to directly tighten them to mitigate feature differences.
In the second step, a naive approach would be to use these concrete bounds to construct a constraint-solving problem, where the parameter changes of the final layer $f_{\mathrm{L}}$ are treated as optimization variables.
However, since concrete bounds are often overly conservative, \name synthesizes symbolic bounds to formulate a more precise problem, thereby avoiding excessive modifications.
To make the new problem solvable by off-the-shelf solvers, we introduce the dual theorem to eliminate nonlinear operations while preserving soundness.
Finally, \name establishes a Mixed-Integer Linear Programming (MILP) problem, ensuring that the solution induces a repaired model that is provably fair over the given set $\mathcal{S}(\bm{x})$.

We have implemented \name as a self-contained toolkit and evaluated it on four popular datasets involving various sensitive attributes.
The results demonstrate its effectiveness in correcting the unfairness over the given set $\mathcal{S}(\bm{x})$ with provable guarantees, and its substantial improvements in generalization over existing state-of-the-art repair methods.
On average, \name achieves 95.93\% and 93.16\% relative fairness improvement on the full dataset and the full input space, where the best baseline achieves only 71.44\% and 72.67\%.
For the setting with multiple sensitive attributes and more practical fairness definition, \name also exhibits remarkable generalization, keeping around 90\% fairness improvement.
Additional experiments with four state-of-the-art fairness testing frameworks further confirm its consistent effectiveness.

To summarize, this paper makes the following contributions:
\begin{itemize}[left=0pt]
    \item We present \name, a novel framework for provable neural network fairness repair with three key technical innovations:
    \begin{enumerate}
        \item We leverage interval bound propagation to soundly capture the model outputs, and design a bounds tightening process to effectively mitigate biased behavior.
        \item We synthesize symbolic bounds to precisely encode fairness constraints and model modifications into a unified constraint-solving formulation with provable guarantees.
        \item We leverage the duality theorem to eliminate nonlinearities and construct an MILP tractable by existing solvers.
    \end{enumerate}
    \item  We evaluate \name on four widely adopted benchmark datasets and demonstrate that it significantly outperforms the state-of-the-art in terms of provable guarantees and generalization to unseen samples.
    \item We release the code and scripts for this paper at \url{https://github.com/nninjn/ProF} to facilitate future studies.
\end{itemize}

\section{Preliminary}
\noindent \textbf{DNNs.} In this work, we focus on DNN models for binary classification tasks. A DNN can be represented as a function \(f: \mathbb{R}^m \to \mathbb{R}\), which maps a high-dimensional input \(\bm{x} \in \mathbb{R}^m\) to an output \(f(\bm{x}) \in \mathbb{R}\). 
It typically consists of an input layer \(f_1\), several hidden layers \(\{f_2, \cdots, f_{\mathrm{L}-1}\}\), and an output layer \(f_\mathrm{L}\).
The first \(\mathrm{L}-1\) layers, denoted as \(f_{1:\mathrm{L}} \eqdef f_{\mathrm{L}-1} \circ \cdots \circ f_1\), 
transform the input into a \( d \)-dimensional feature space, i.e., \( f_{1:\mathrm{L}}(\bm{x}) \in \mathbb{R}^d \),
while the final layer \(f_\mathrm{L}\) makes the classification.
The model predicts the positive class if \( f(\bm{x}) \geq 0 \) and the negative class otherwise.
Given a training dataset $\mathcal{D}_{\text{train}} $, a DNN is trained by minimizing binary cross-entropy (BCE) loss $\ell$: 
\[
\min_f \frac{1}{|\mathcal{D}_{\text{train}}|} \sum_{(\bm{x}^i,y^i) \in \mathcal{D}_{\text{train}}} \ell(\sigma( f(\bm{x}^i)), y^i)
\]
where $\sigma$ and $y^i \in \{0, 1\}$ denote the sigmoid function and the true label for the input $\bm{x}^i$, respectively.

\noindent \textbf{Interval Bound Propagation.} 
Owing to its efficiency in computing NN output ranges, interval bound propagation is widely used in NN verification~\cite{deeppoly, zhang2018efficient, gehr2018ai2, paulsen2022example, wang2018formal, deepsrgrextended}. 
The core idea involves propagating interval bounds layer-wise through the model, starting from predefined input domains.
In contrast to exact verification methods~\cite{bunel2020branch, marabou, katz2017reluplex} that rely on SAT solvers, interval arithmetic offers scalable approximations by symbolizing each neuron's value as an interval derived from its predecessor layer.
Here we briefly describe how it propagates the bounds through the linear layer and the activation layer (using ReLU as an example). 
For a neuron \(\mathrm{h}\) after a linear transformation with weights \(\mathbf{W}\) and bias \(b\), its output bounds are computed as:
\begin{equation}
\label{eq:ibp-linear}
\begin{split}
\underline{\mathrm{h}} &= b + \sum_{i} \left( \max(0, \mathrm{W}_i) \cdot \underline{\mathrm{z}}_i + \min(0, \mathrm{W}_i) \cdot \overline{\mathrm{z}}_i \right), \\
\overline{\mathrm{h}} &= b + \sum_{i} \left( \max(0, \mathrm{W}_i) \cdot \overline{\mathrm{z}}_i + \min(0, \mathrm{W}_i) \cdot \underline{\mathrm{z}}_i \right)
\end{split}
\end{equation}
where \(\underline{\mathrm{z}}_i, \overline{\mathrm{z}}_i\) are the bounds of the \(i\)-th neuron in the predecessor layer. 
For an unstable ReLU neuron \(\mathrm{h} = \mathrm{ReLU}(\mathrm{z})\) with input interval \([l, u]\) where \(l < 0 < u\), a sound interval propagation
can be derived via linear relaxation in various ways~\cite{zhang2018efficient, wang2018efficient}. 
Here we show a simple triangular form:
\begin{equation}
\label{eq:ibp-relu}
0 \leq \mathrm{h} \leq \frac{u}{u - l} \cdot (\mathrm{z} - l ) , \;
\left[\underline{\mathrm{h}}, \overline{\mathrm{h}}\right]=[0, \, u]
\end{equation}

\begin{figure}[t]
\centering
\includegraphics[width=1.0\linewidth]{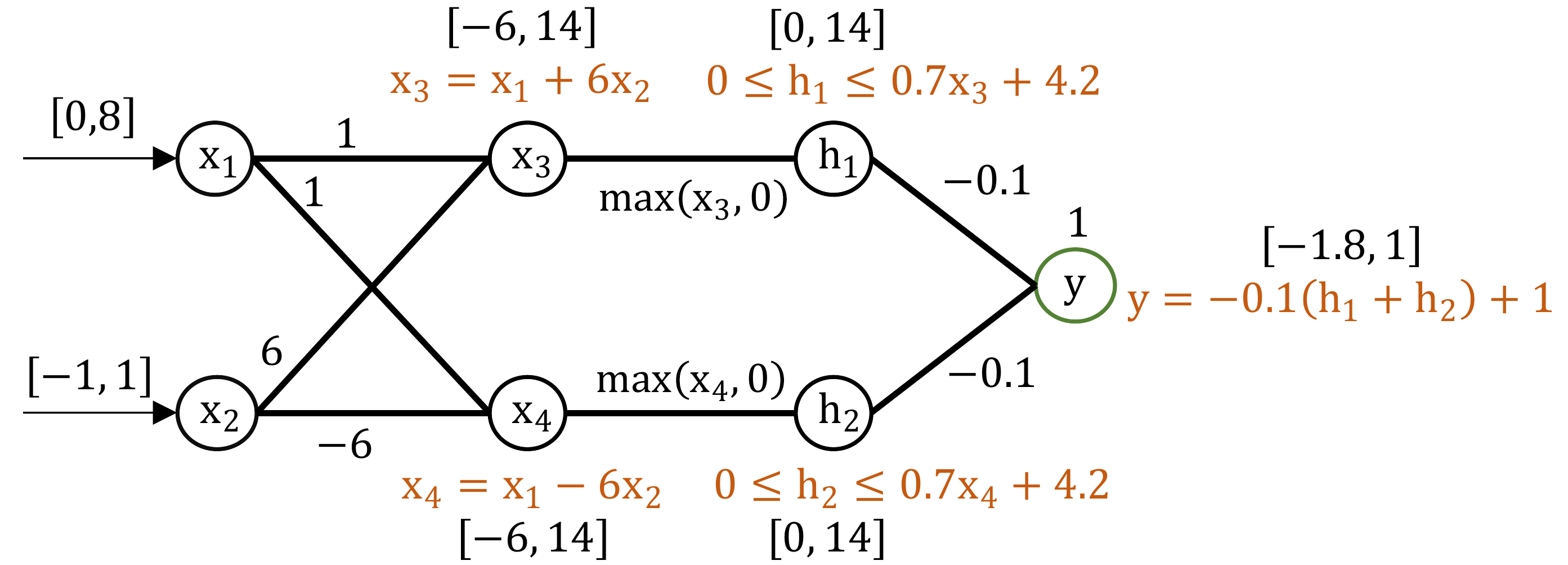}
\caption{Interval bound propagation on an example NN. All neurons have zero bias except the output neuron (bias=1).}
\label{fig:ibp-example}
\end{figure}

\begin{figure*}[!t]
\centering
\includegraphics[width=0.96\textwidth]{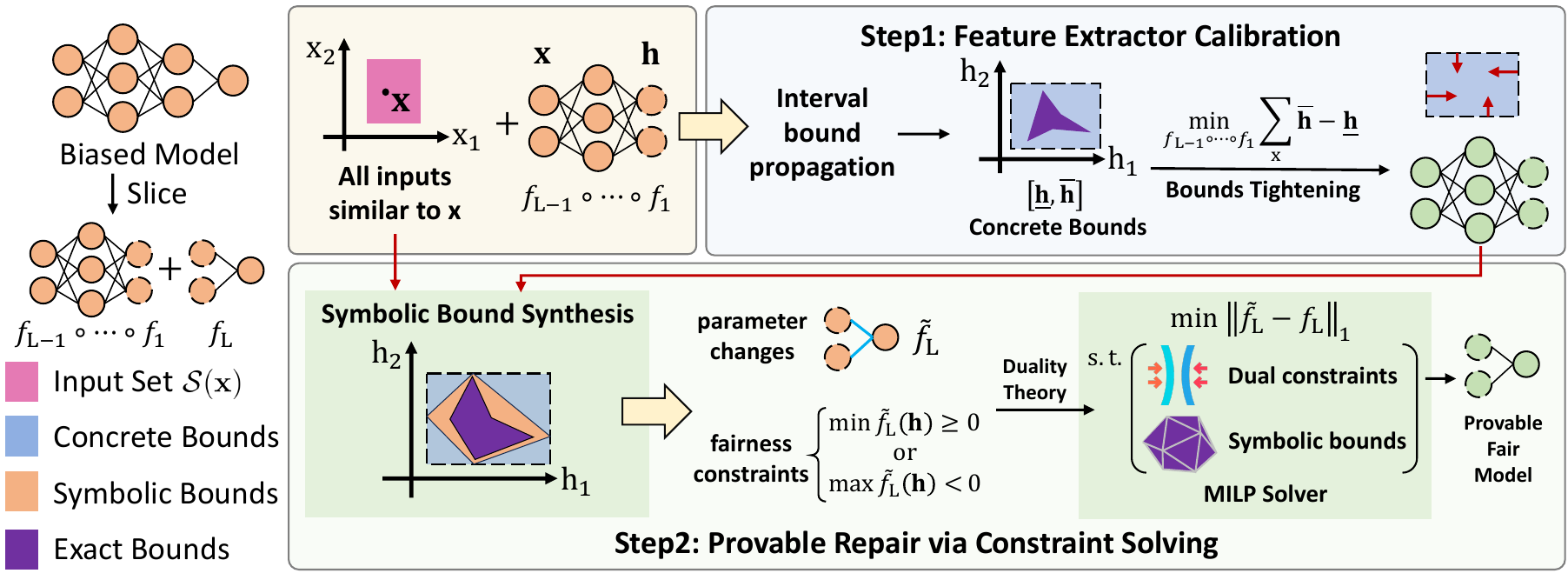}
\caption{The overview of \name framework. For notational simplicity, we use \( f_{1:\mathrm{L}} \) to denote \( f_{\mathrm{L}-1} \circ \cdots \circ f_1 \) hereafter.}
\label{fig:overview}
\end{figure*}

\begin{example}
\label{exam:ibp}
Fig.~\ref{fig:ibp-example} illustrates interval bound propagation on a simple NN, where \textcolor[HTML]{C55A11}{the brown equations} denote the neurons' symbolic bounds and the black square brackets are concrete bounds.
The input \(x_1\) and \(x_2\) are bounded by \([0, 8]\) and \([-1, 1]\).  
Using Eq.~\eqref{eq:ibp-linear}, the interval propagation for the first layer concretizes the symbolic expressions $x_3 = x_1 + 6x_2$ and $x_4 = x_1 - 6x_2$ to concrete bounds $[-6,14]$ and $[-6,14]$.
Then the triangular relaxation from Eq.~\eqref{eq:ibp-relu} is used to calculate the bounds of \(h_1\) and \(h_2\), refining both to \([0,14]\).
The output \(y\) finally yields the bound \([-1.8, 1]\). 
Notably, this bound is usually conservative due to the input dependencies~\cite{de2004affine, moore2009introduction} of interval arithmetic, which assumes that the propagation of each neuron is independent.
For example, the minimum \(y=-1.8\) implies \(h_1=h_2=14\), where no valid input \((x_1,x_2)\) can simultaneously achieve \(h_1=14\) and \(h_2=14\).
We analyze how this issue impacts the repair process and present our solution in Sec.~\ref{sec:concrete-repair} and Sec.~\ref{sec:symbolic-repair}.
\end{example}

\noindent \textbf{Individual Fairness.}
As established in prior work~\cite{zhang2020white, zheng2022neuronfair, chen2024fairness}, individual fairness (IF) requires a model to produce consistent predictions for similar individuals who differ only in sensitive attributes. 
Let \(\bm{x} = (x_1, x_2, \cdots, x_m) \in \mathbb{R}^m\) be an input. 
We define \(\mathcal{S}(\bm{x})\) as the set of all inputs similar to \(\bm{x}\): 
\[ \mathcal{S}(\bm{x}) =
\left\{ \bm{x}^{\prime} \mid \exists j \in \mathrm{P}, \, x_{j} \neq x^{\prime}_j ; \, \forall j \in \mathrm{NP}, \, x_{j} = x^{\prime}_j \right\}\]
where \(x_j\) and \(x^{\prime}_j\) denote the values of the $j$-th attribute in $\bm{x}$ and $\bm{x}^{\prime}$, \(\mathrm{P}\) is the set of sensitive/protected attribute indices, and \(\mathrm{NP}\) represent the set of non-sensitive attribute indices.
Recent work~\cite{biswas2023fairify, john2020verifying} further relaxes this notion by allowing small perturbations on non-sensitive attributes. Specifically, a relaxed neighborhood \(\mathcal{S}(\bm{x})\) is defined as:
\[
\mathcal{S}(\bm{x}) =
\left\{ \bm{x}^{\prime} \mid \exists j \in \mathrm{P}, \, x_{j} \neq x^{\prime}_j ; \, \forall j \in \mathrm{NP}, \, |x_j - x^{\prime}_j| \leq \epsilon_j \right\}
\]
where \(\epsilon_j > 0\) limits allowable variations on non-sensitive attribute $j$, reflecting that small differences (e.g., age differences of a few years) may not violate the fairness relation.

Building on above notions, we define an input \(  \bm{x} \) as an Individual Discriminatory Instance (IDI) if there exists \( \bm{x}^{\prime}\in \mathcal{S}(\bm{x}) \) such that:
\(
\left(f(\bm{x}) \ge 0 \wedge f(\bm{x}^{\prime}) < 0 \right) \, \vee \, \left(f(\bm{x}) < 0 \wedge f(\bm{x}^{\prime}) \ge 0 \right)
\).
Such a pair \(\langle \bm{x}, \bm{x}^{\prime} \rangle\) is referred to as an IDI pair.
Given a NN $f$, we further define \(\mathcal{U}(f, \bm{x}, \mathcal{S}(\bm{x}) )\)  as the set of all inputs in \(\mathcal{S}(\bm{x}) \) that result in unfair classification under \(f\):
\begin{equation}
\label{eq:U-set}
\begin{aligned}
\mathcal{U}(f,\bm{x},\mathcal{S}(\bm{x}) ) =\{ \bm{x}^\prime \in \mathcal{S}(\bm{x}) \mid(f(\bm{x})  \geq 0 \wedge f(\bm{x}^\prime) < 0 ) \\ 
\vee (f(\bm{x})  < 0 \wedge f(\bm{x}^\prime) \geq 0 )\} 
\end{aligned}
\end{equation}

\begin{example}
\label{exam:fairness}
Let us revisit Fig.~\ref{fig:ibp-example}.
Suppose we have an input \(\bm{x} = (4, 0)\) with \(f(\bm{x}) = 0.2\), where \(x_1\) is the protected attribute ranging over \([0, 8]\), and \(x_2\) is a non-sensitive attribute with \(\epsilon = 1\). The neighborhood \(\mathcal{S}(\bm{x})\) is thus defined as \(\{\bm{x}' \mid 0 \le x_1' \le 8,\ -1 \le x_2' \le 1\}\).  
It is easy to verify that there exists \(\bm{x}' \in \mathcal{S}(\bm{x})\) such that \(f(\bm{x}') < 0\); for example, \(\bm{x}' = (8, 1)\) yields \(f(\bm{x}') = -0.6\).  
Therefore, \(\bm{x}\) constitutes an IDI, and \(f\) violates individual fairness at this input.
\end{example}

\noindent \textbf{Problem Formulation.} We now formalize our repair problem.
\begin{definition}[\textbf{The IF repair problem}]
Given a DNN \(f\) and a repair set \(\mathcal{D}_{\text{r}}=\{\bm{x}^1, \bm{x}^2, \ldots, \bm{x}^n\}\), where \(\bm{x}^i\) denotes the $i$-th input. The goal of provable repair is to find a modified DNN \(\tilde{f}\) that guarantees IF over $\mathcal{S}(\bm{x}^i)$ for every \(\bm{x}^i \in \mathcal{D}_{\text{r }}\), that is:
\begin{equation}
\begin{aligned}
\forall \bm{x}^i \in \mathcal{D}_{\text{r}}, \; \forall\bm{x} \in \mathcal{S}(\bm{x}^i): \mathcal{U}(\tilde{f},\bm{x},\mathcal{S}(\bm{x}^i) )  = \emptyset
\\
\end{aligned}
\end{equation}
\end{definition}
\noindent We also aim for the repair to exhibit generalization by mitigating unfairness on inputs beyond $\mathcal{D}_{\text{r}}$. 
Consistent with prior repair work~\cite{chen2024isolation, sun2022causality}, we assume access to a small set $\mathcal{D}_{\text{c}}$ of training data to preserve the model's original performance.

\section{Methodology}
\subsection{Overview}
Fig.~\ref{fig:overview} presents \name, our provable fairness repair framework for DNNs.
It consists of two main components: (1) unfair feature extractor calibration via \textbf{concrete bounds tightening}, and (2) provable repair through \textbf{symbolic bounds synthesis and constraint-solving}.
The first step operates in the feature space: by propagating interval bounds through the network, \name establishes sound over-approximation for the features of all similar individuals. 
Then a progressively bounds tightening process is designed to rectify the biased feature extractor, promoting feature consistency on similar individuals.
In the second step, symbolic bounds (which are tighter than concrete bounds) are synthesized and combined with the fairness constraints to formulate a more precise constraint-solving problem that provides guarantees for repair. 
To eliminate the nonlinear operations, we introduce the dual theorem and transform the original problem into an MILP, while preserving soundness.

\subsection{Correcting Biased Feature Extractor}
The objective of this phase is to calibrate the model's feature extraction part $f_{1:\mathrm{L}}$ so that it produces features as consistent as possible for all inputs in the similar set $\mathcal{S}$.
The ideal scenario is that for any inputs $\bm{x}, \bm{x}' \in \mathcal{S}(\bm{x}^i)$, $f_{1:\mathrm{L}}$ produces identical features and thus $f$ makes the same classification, i.e.,
\begin{equation}
\label{eq:ideal}
    \forall \bm{x}, \bm{x}^{\prime}\in\mathcal{S}(\bm{x}^i): f_{1: \mathrm{L}}(\bm{x}) = f_{1: \mathrm{L}}(\bm{x}^{\prime})
\end{equation}
Some prior works~\cite{li2024runner, chen2024isolation, quan2025dissecting} pursue similar goals by either: (1) 
biased neuron identification via differential analysis on given IDI pairs followed by parameter modification, or (2) retraining the models to minimize the distances between features extracted from the IDI pairs.
In other words, they aim to reduce the distances between specific feature pairs (discrete point pairs in the purple region in Fig.~\ref{fig:overview}, Step 1).
However, these approaches rely on IDI pair $\langle \bm{x}, \bm{x}^{\prime} \rangle$ generated either by sampling within the similar input set $\mathcal{S}(\bm{x})$ or existing testing methods~\cite{zhang2020white,wang2024maft}.
This sampling-driven nature inherently suffers from incomplete coverage of $\mathcal{S}(\bm{x})$, as finite samples cannot exhaustively represent the entire set. 
Consequently, some uncovered instances may still retain large feature discrepancies and thus lead to unfair classifications.

To address this problem, we propose utilizing interval bound propagation to soundly \emph{characterize} and \emph{mitigate} the feature differences across the neighborhood set $\mathcal{S}(\bm{x})$.
Unlike discrete sampling methods, our approach over-approximates the model's feature-space outputs. 
Specifically, we employ auto-LiRPA~\cite{xu2020automatic} to calculate\footnote{There are several other interval propagation tools, such as DeepPoly~\cite{deeppoly} and DeepZ~\cite{deepz}, which vary in precision. 
We choose auto-LiRPA~\cite{xu2020automatic} due to its balanced trade-off between efficiency and precision.} the concrete bounds $[\underline{\mathbf{h}}^i, \overline{\mathbf{h}}^i] \subset \mathbb{R}^d$ for each set $\mathcal{S}(\bm{x}^i)$, such that:
\begin{equation}
\forall i \in [n] , \, \bm{x} \in \mathcal{S}(\bm{x}^i): \; f_{1: \mathrm{L}}(\bm{x}) \in \left[\underline{\mathbf{h}}^i, \overline{\mathbf{h}}^i \right]
\end{equation}
where $[n]$ denotes the set $\{1, 2, \dots ,n\}$.
As shown in Fig.~\ref{fig:overview} (Step 1, blue region), the derived interval $[\underline{\mathbf{h}}^i, \overline{\mathbf{h}}^i]$ forms an axis-aligned hyperrectangle that provably encloses the exact feature set $\mathcal{H}^i \eqdef \{f_{1:\mathrm{L}}(\bm{x}) \mid \bm{x} \in \mathcal{S}(\bm{x}^i)\}$ (purple region).
The axis-aligned property enables us to upper bound the maximum distance $L_{dis}$ between features in $\mathcal{H}^i$:
\begin{equation}
\begin{aligned}
L_{dis} \eqdef &\max_{\substack{\bm{x}, \,  \bm{x}^{\prime} \in \mathcal{S}(\bm{x}^i)}}  \|f_{1: \mathrm{L}}(\bm{x}) - f_{1: \mathrm{L}}(\bm{x}^{\prime})\|_1 \\
=& \max_{\mathbf{h},\mathbf{h}^{\prime} \in \mathcal{H}^i} \|\mathbf{h}-\mathbf{h}^{\prime}\|_1
\leq \sum_{j=1}^d \max_{\mathbf{h},\mathbf{h}^{\prime} \in \mathcal{H}^i}  |\mathbf{h}_j-\mathbf{h}^{\prime}_j| \\
\leq  & \|\overline{\mathbf{h}}^i - \underline{\mathbf{h}}^i  \|_{1}
\end{aligned}
\end{equation}
With the derived upper bounds in hand, we design a progressive bounds tightening process to mitigate the bias in $f_{1: \mathrm{L}}$.
This procedure iteratively reduces the $\ell_1$-norm of the concrete bounds, thereby contracting the maximum feature distance.
As detailed in Alg.~\ref{algorithm:feature-correct} (Lines 2–4), the process begins by computing initial concrete bounds for each input $\bm{x}^i \in \mathcal{D}_{\text{r}}$ over its neighborhood set $\mathcal{S}(\bm{x}^i)$. 
It then updates the concrete bounds and minimizes a normalized objective, which measures the relative $\ell_1$-norm reduction between current and original bounds.
This ratio, bounded within $[0,1]$, directly quantifies the degree of bias mitigation, where zero corresponds to the ideal scenario described in Eq.~\eqref{eq:ideal}.
We aggregate these ratios across all inputs in $\mathcal{D}_{\text{r}}$ to construct the fairness loss $\mathcal{L}_{\text{fair}}$ (Line 8 of Alg.~\ref{algorithm:feature-correct}).
To preserve the model's performance, the BCE loss $\mathcal{L}_{\text{bce}}$ is computed on a small set of training data $\mathcal{D}_{\text{c}}$.
We optimize the model parameters to minimize the combined objective $\mathcal{L} = \mathcal{L}_{\text{fair}} + \mathcal{L}_{\text{bce}}$ through gradient descent.

\begin{algorithm}[t]
\small
\DontPrintSemicolon
\caption{Biased Feature Calibration}
\label{algorithm:feature-correct}
\SetKw{Parameter}{Parameter:}
\KwIn{Biased NN $f=f_{\mathrm{L}} \circ f_{1:\mathrm{L}}$, repair set \(\mathcal{D}_{\text{r}}\), a small set of training data \(\mathcal{D}_{\text{c}}\), maximum iteration $\mathrm{T}$
} 
\KwOut{Repaired feature extractor $\tilde{f}_{1:\mathrm{L}} $}

\For{$i \gets 1$ \KwTo $|\mathcal{D}_{\text{r}}|$}{
    $\mathcal{S}(\bm{x}^i) \eqdef $ Set of all inputs similar to $\bm{x}^i$
    
    $\left[ \underline{\mathbf{h}}^i, \overline{\mathbf{h}}^i  \right] \gets \textsc{ConcreteBounds}(f_{1:\mathrm{L}}, \mathcal{S}(\bm{x}^i))$

    $\mathrm{OriDiff}_i \gets \left\| \overline{\mathbf{h}}^i - \underline{\mathbf{h}}^i \right\|_1$
}

\For{$iter\leftarrow 1$ \KwTo $\mathrm{T}$}{

    \For{$i \gets 1$ \KwTo $|\mathcal{D}_{\text{r}}|$}{
        $\left[\underline{\mathbf{h}}^i, \overline{\mathbf{h}}^i \right] \gets \textsc{ConcreteBounds}(f_{1:\mathrm{L}}, \mathcal{S}(\bm{x}^i))$
        
        \mytcc{Update the concrete bounds}
    }

    $\mathcal{L}_{\text{fair}} \gets \frac{1}{\left|\mathcal{D}_{\text{r}}\right|} \sum_{\substack{i=1}}^{|\mathcal{D}_{\text{r}}|} \dfrac{\left \|\overline{\mathbf{h}}^i - \underline{\mathbf{h}}^i \right \|_{1}}{\mathrm{OriDiff}_i}$  

    \vspace{0.2em}
    $\mathcal{L}_{\text{bce}} \gets \frac{1}{\left|\mathcal{D}_{\text{c}}\right|} \sum_{\substack{(\bm{x},y)\in\mathcal{D}_{\text{c}}}} 
    \ell(\sigma( f(\bm{x})), y)$
    \vspace{0.2em}

    $f_{1:\mathrm{L}} \gets \textsc{GradientDescent}(f_{1:\mathrm{L}}; \mathcal{L}_{\text{fair}} + \mathcal{L}_{\text{bce}})$ 
}

\Return $f_{1:\mathrm{L}}$
\end{algorithm}

\noindent \textbf{Remark:} While the above process can effectively mitigate bias in the feature extractor (see Section~\ref{sec:exp-rq4}), the fixed-number gradient descent steps do not guarantee provable repair. 
The repaired feature extractor may still yield $\mathcal{L}_{\text{fair}} > 0$, indicating potential feature discrepancies within some $\mathcal{S}(\bm{x}^i)$ that could lead to unfair classifications. 
To address this problem, we proceed to repair the final layer $f_{\mathrm{L}}$ in the following sections. We first present a naive method that directly utilizes the concrete bounds, followed by our proposed advanced approach that incorporates symbolic bounds for enhanced precision.

\subsection{A Naive Provable Repair Method with Concrete Bounds}
\label{sec:concrete-repair}
In this section, we present a naive repair method for the final classification layer $f_{\mathrm{L}}$.
Our core idea is to formulate the repair as a constraint-solving problem and leverage existing solvers to obtain the solution with provable guarantees.
Formally, the fairness repair on $f_{\mathrm{L}}$ is formulated as follows:
\begin{align}
&\min_{\Delta \mathbf{W}, \Delta \mathrm{b}} \|\Delta \mathbf{W}\|_1 + \|\Delta \mathrm{b}\|_1\label{eq:min-obj} \\
\text{s.t.} \; & \forall i \in [n], \; \mathcal{H}^i = \{\tilde{f}_{1:\mathrm{L}}(\bm{x}) \mid \bm{x} \in \mathcal{S}(\bm{x}^i)\} \label{eq:exact-H}\\
&    \forall i \in [n], \; \min_{\substack{\mathbf{h} \in \mathcal{H}}^i } \tilde{f}_\mathrm{L}(\mathbf{h}) \geq 0 \vee \max_{\substack{\mathbf{h} \in \mathcal{H}}^i} \tilde{f}_\mathrm{L}(\mathbf{h}) < 0 \label{eq:output-IF} 
\end{align}
Here, $\tilde{f}_{1:\mathrm{L}}$ denotes the feature extractor calibrated by Alg.~\ref{algorithm:feature-correct} and $\tilde{f}_\mathrm{L}(\mathbf{h}) = (\mathbf{W} + \Delta\mathbf{W})\mathbf{h} + \mathrm{b} + \Delta \mathrm{b}$ is the repaired final layer, where $\mathbf{W} + \Delta\mathbf{W} \in \mathbb{R}^{1 \times d}$ and $\mathbf{h} \in \mathbb{R}^{d}$.
Eq.~\eqref{eq:min-obj} minimizes parameter modifications while Eq.~\eqref{eq:output-IF} and Eq.~\eqref{eq:exact-H} enforce fairness through output consistency across all inputs in $\mathcal{S}(\bm{x}^i)$.

We identify that existing solvers cannot directly handle the formulated problem due to two inherent complexities: (1) $\tilde{f}_{1:\mathrm{L}}$ includes multiple nonlinear layers and thus the exact feature set $\mathcal{H}^i$ is highly non-convex, and (2) the constraint in Eq.~\eqref{eq:output-IF} introduces multiplication terms $\Delta \mathbf{W}\mathbf{h}$, which are non-linear.
To address the former, a straightforward method is to use the concrete bounds again to over-approximate $\mathcal{H}^i$.
Specifically, we simplify Eqs.~\eqref{eq:exact-H} and~\eqref{eq:output-IF} as follows:
\begin{equation}
\label{eq:concrete-fair-cons}
\begin{aligned}
\forall i \in [n], \;& \mathrm{b} + \Delta \mathrm{b} + \min_{\mathbf{h} \in \left[\underline{\mathbf{h}}^i, \overline{\mathbf{h}}^i\right] }  (\mathbf{W} + \Delta \mathbf{W}) \mathbf{h} \geq 0 \\ \vee \; & \mathrm{b} + \Delta \mathrm{b} +\max_{\mathbf{h} \in \left[\underline{\mathbf{h}}^i, \overline{\mathbf{h}}^i \right] }  (\mathbf{W} + \Delta \mathbf{W}) \mathbf{h}  < 0\
\end{aligned}
\end{equation}
Since $[\underline{\mathbf{h}}^i, \overline{\mathbf{h}}^i ]$ over-approximates $\mathcal{H}^i$, Eq.\eqref{eq:concrete-fair-cons} consists of strengthened constraints that imply Eqs.\eqref{eq:output-IF} and~\eqref{eq:exact-H}.
Moreover, the above simplification enables us to introduce two sets of variables $\{\mathrm{P}_{i,j}\}_{j \in [d]}$ and $\{\mathrm{Q}_{i,j}\}_{j \in [d]}$ to capture the extreme values of $(\mathbf{W} + \Delta \mathbf{W}) \mathbf{h}$ in each dimension:
\begin{equation}
\label{eq:PQ}
\begin{aligned}
\mrm{P}_{i,j} \le (\mathbf{W}_j + \Delta \mathbf{W}_j) \underline{\mathbf{h}}^i_j 
\wedge \mrm{P}_{i,j} \le (\mathbf{W}_j + \Delta \mathbf{W}_j) \overline{\mathbf{h}}^i_j \\
\mrm{Q}_{i,j} \ge (\mathbf{W}_j + \Delta \mathbf{W}_j) \underline{\mathbf{h}}^i_j 
\wedge \mrm{Q}_{i,j} \ge (\mathbf{W}_j + \Delta \mathbf{W}_j) \overline{\mathbf{h}}^i_j 
\end{aligned}
\end{equation}
where $\underline{\mathbf{h}}_j^i$ and $\overline{\mathbf{h}}_j^i$ are the $j$-th dimensional endpoints of the hyperrectangle $[\underline{\mathbf{h}}^i, \overline{\mathbf{h}}^i ]$. 
Therefore, the summations over $j$ of $\mathrm{P}_{i,j}$ and $\mathrm{Q}_{i,j}$ can provide lower and upper bounds for $(\mathbf{W} + \Delta \mathbf{W}) \mathbf{h}$ across the whole hyperrectangle $[\underline{\mathbf{h}}^i, \overline{\mathbf{h}}^i]$ while avoiding the multiplication between $\mathbf{h}$ and $\Delta \mathbf{W}$, i.e.,
\begin{equation}
\begin{aligned}
\sum_{j\in [d]} \mrm{P}_{i,j} \leq \min_{\mathbf{h} \in \left[\underline{\mathbf{h}}^i, \overline{\mathbf{h}}^i \right] } 
(\mathbf{W} + \Delta \mathbf{W}) \mathbf{h}  \\
\sum_{j\in [d]} \mrm{Q}_{i,j} \geq \max_{\mathbf{h} \in \left[\underline{\mathbf{h}}^i, \overline{\mathbf{h}}^i \right] } 
(\mathbf{W} + \Delta \mathbf{W}) \mathbf{h}
\end{aligned}
\end{equation}
Let $\mathrm{LB}_i \eqdef \mathrm{b} + \Delta \mathrm{b} + \sum_{j\in [d]} \mrm{P}_{i,j}$ and $\mathrm{UB}_i \eqdef \mathrm{b} + \Delta \mathrm{b} + \sum_{j\in [d]} \mrm{Q}_{i,j}$, we then transform the repair problem as:
\begin{equation}
\label{eq:concrete-repair}
\begin{aligned}
&\min_{\Delta \mathbf{W}, \Delta \mathrm{b}}  \|\Delta \mathbf{W}\|_1 + \|\Delta \mathrm{b}\|_1\ \\ 
\text{s.t.} \; & \forall i \in [n]    ~\eqref{eq:PQ}, \;    \mathrm{LB}_i \geq 0 \vee \mathrm{UB}_i <0
\end{aligned}
\end{equation}
Note that the disjunction ($\vee$) can be handled using Big-M method (detailed in the next section).
The upper half of Fig.~\ref{fig:symbolic-moti} revisits Example~\ref{exam:ibp}, where the \textcolor[HTML]{2F5597}{blue} values indicate the parameter modifications from solving Problem~\eqref{eq:concrete-repair}, which enforces $\forall \mathrm{h}_1 \in [0,14], \mathrm{h}_2 \in[0,14]: (\frac{9}{140}-0.1)\times\mathrm{h}_1+(\frac{9}{140}-0.1) \times \mathrm{h}_2+1 \geq 0$ to ensure fairness.
The provable guarantees achieved by solving Problem~\eqref{eq:concrete-repair} are formally stated as:
\begin{theorem}
\label{thm:main-c}
Let $\mathcal{D}_{\text{r}}=\{\bm{x}^i\}_{i=1}^n$ be a set of inputs, each associated with a similarity neighborhood $\mathcal{S}(\bm{x}^i)$, and let $f=f_{\mathrm{L}}\circ \tilde{f}_{1:\mathrm{L}}$ be a DNN. 
Then any feasible solution $(\Delta \mathbf{W}, \Delta \mathrm{b})$ to the problem~\eqref{eq:concrete-repair} induces a repaired final layer $\tilde{f}_{\mathrm{L}}$ such that the composite model $\tilde{f}=\tilde{f}_{\mathrm{L}}\circ \tilde{f}_{1:\mathrm{L}}$ provably satisfies individual fairness on all $\mathcal{S}(\bm{x}^i)$, $ i \in [n]$.
\end{theorem}
\begin{proof}
Consider any feasible solution $(\Delta \dot{\mathbf{W}}, \Delta \dot{\mathrm{b}})$ to problem~\eqref{eq:concrete-repair}.
Recall that the sets $\mathcal{H}^i$ satisfy
$
\mathcal{H}^i = \{ f_{1:\mathrm{L}}(\bm{x}) \mid \bm{x} \in \mathcal{S}(\bm{x}^i) \} \subseteq [\underline{\mathbf{h}}^i, \overline{\mathbf{h}}^i]
$
and that the lower and upper bounds $\mathrm{LB}_i, \mathrm{UB}_i$ are constructed via the extreme values of auxiliary optimization problems $\mathrm{P}_{i,j}$ and $\mathrm{Q}_{i,j}$ to satisfy
\begin{equation*}
\begin{aligned}
    \mathrm{LB}_i \leq \min_{\mathbf{h} \in [\underline{\mathbf{h}}^i, \overline{\mathbf{h}}^i]} (\mathbf{W} + \Delta \dot{\mathbf{W}}) \mathbf{h} + \mathrm{b} + \Delta \dot{\mathrm{b}} \\
 \mathrm{UB}_i \geq \max_{\mathbf{h} \in [\underline{\mathbf{h}}^i, \overline{\mathbf{h}}^i]} (\mathbf{W} + \Delta \dot{\mathbf{W}}) \mathbf{h} + \mathrm{b} + \Delta \dot{\mathrm{b}}
\end{aligned}
\end{equation*}

Since $\mathcal{H}^i \subseteq [\underline{\mathbf{h}}^i, \overline{\mathbf{h}}^i]$, these bounds also hold over $\mathcal{H}^i$. By feasibility, the repair constraints ensure that either
$\mathrm{LB}_i \geq 0  \text{ or }  \mathrm{UB}_i < 0$, which guarantees that for all $\mathbf{h} \in \mathcal{H}^i$, the repaired model output $\tilde{f}_{\mathrm{L}}(\mathbf{h})$ is consistently non-negative or negative.
Consequently, the composite model $\tilde{f} = \tilde{f}_{\mathrm{L}} \circ \tilde{f}_{1:\mathrm{L}}$ provably satisfies individual fairness across all neighborhoods $\mathcal{S}(\bm{x}^i)$, $i \in [n]$.\qedhere
\end{proof}

\subsection{Provable Repair with Symbolic Bounds}
\label{sec:symbolic-repair}
We now identify a crucial problem of the proposed naive method: the concrete bound $[\underline{\mathbf{h}}^i, \overline{\mathbf{h}}^i]$ is typically too conservative to capture the exact feature set $\mathcal{H}^i$ (\textcolor[HTML]{7030A0}{purple} region in figure).
This leads to unnecessary modifications of the final layer 
in order to enforce fairness across the entire box $[\underline{\mathbf{h}}^i, \overline{\mathbf{h}}^i ]$ rather than just the exact set $\mathcal{H}^i$.
As illustrated in the upper half of Fig.~\ref{fig:symbolic-moti}, such conservativeness results in excessive changes that account for spurious points outside $\mathcal{H}^i$. 
This issue is exacerbated when the repair set $\mathcal{D}_{\text{r}}$ contains multiple inputs, requiring simultaneous satisfaction of fairness constraints over multiple sets $\mathcal{S}(\bm{x})$ (see Section~\ref{sec:exp-rq4}).

\begin{figure}[t]
\centering
\includegraphics[width=1.0\linewidth]{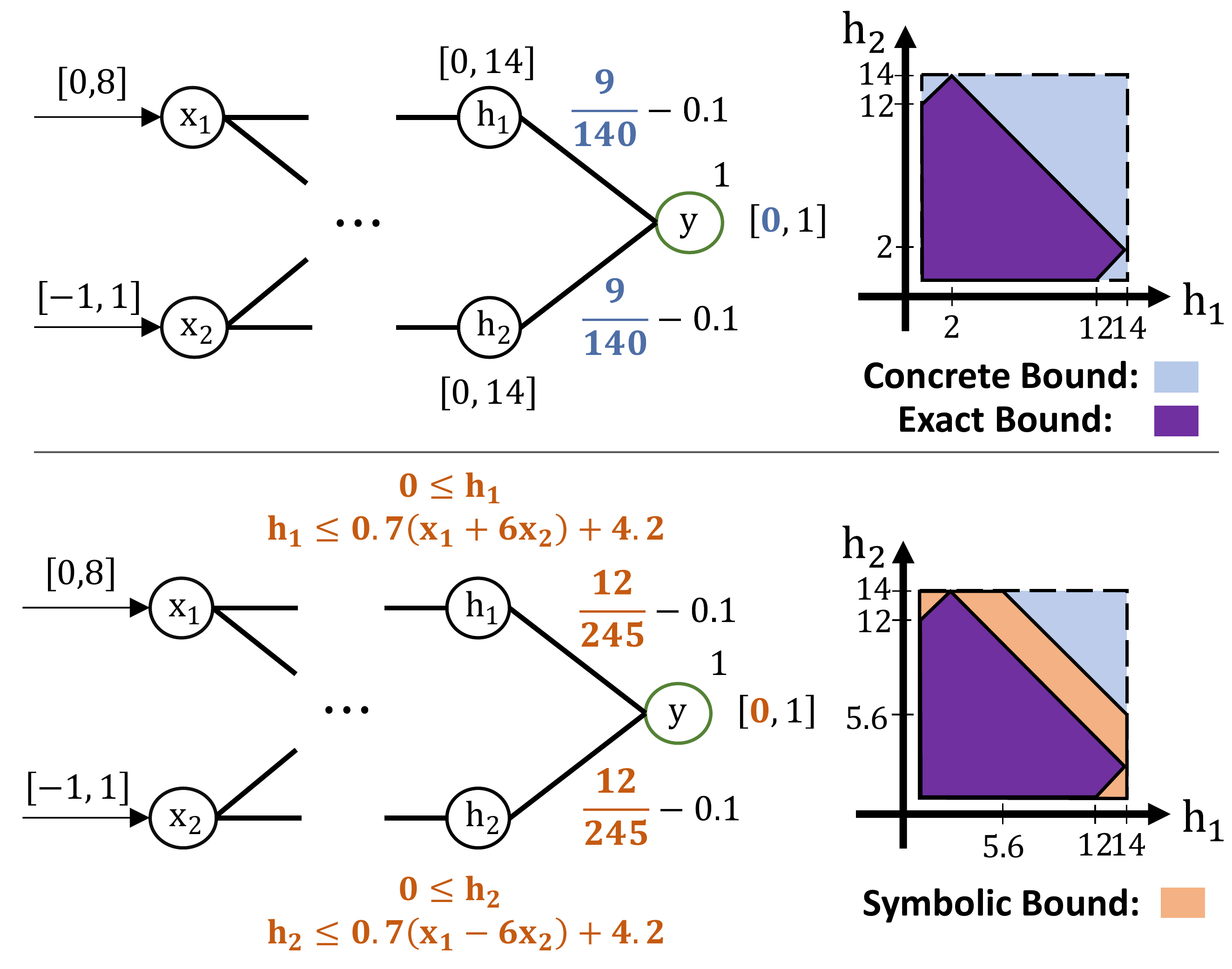}
\caption{
Comparison of repair with concrete bounds and symbolic bounds. 
Purple region shows the true feature set $\mathcal{H}^i$.}
\label{fig:symbolic-moti}
\end{figure}

As discussed in Example~\ref{exam:ibp}, the imprecision of concrete bounds fundamentally stems from their assumption of inter-neuron independence during propagation, which discards crucial variable dependencies. 
To capture $\mathcal{H}^i$ more precisely, we propose synthesizing \emph{symbolic bounds}, where each neuron's bounds are characterized as a linear combination of preceding variables, thereby preserving the dependencies.
Given a set $\mathcal{S}(\bm{x}^i)$, we first utilize auto-LiRPA (other bound propagation tools are also applicable) to calculate the linear bounds between $\mathbf{h}$ and $\bm{x}$, and further incorporate the input bounds to synthesize the symbolic bounds for repair as: 
\begin{equation}
\label{eq:def-hat-H}
\begin{aligned}
\widehat{\mathcal{H}}^i \eqdef \left\{ \mathbf{h} \,\middle|\, 
\bm{\alpha}_{i}^{\le}  \bm{x}+\bm{\beta}_{i}^{\le} \leq \mathbf{h} \leq \bm{\alpha}_{i}^{\ge}  \bm{x}+\bm{\beta}_{i}^{\ge}, \right. \\
\left. \mathcal{S}_{L}(\bm{x}^i) \leq \bm{x} \leq \mathcal{S}_{U}(\bm{x}^i) \right\}
\end{aligned}
\end{equation}
where $\bm{\alpha}_{i}^{\le}$, $\bm{\beta}_{i}^{\le}$, $\bm{\alpha}_{i}^{\ge}$ and $\bm{\beta}_{i}^{\ge}$ are linear coefficients and biases, $\mathcal{S}_{L}(\bm{x}^i)$ and $\mathcal{S}_{U}(\bm{x}^i)$ denote the lower bound and upper bound of $\mathcal{S}(\bm{x}^i)$, respectively.
\textcolor[HTML]{C55A11}{The brown equations} in Fig.~\ref{fig:symbolic-moti} denote the symbolic bounds for our running example:
\begin{equation*}
\begin{aligned}
\widehat{\mathcal{H}}^i = \left\{ \mathbf{h} \,\middle|\, 
  0 \leq \mathrm{h}_1 \leq 0.7(x_1 + 6x_2) + 4.2, 0 \leq x_1 \leq 8\right.\\
 \left.  0 \leq \mathrm{h}_2 \leq 0.7(x_1 - 6x_2) + 4.2  , -1 \leq x_2 \leq 1  \right\}
\end{aligned}
\end{equation*}
We can see that the region formed by $\widehat{\mathcal{H}}^i$ (\textcolor[HTML]{C55A11}{the brown area}) is more precise than the one formed by concrete bounds.

Once the symbolic bounds is obtained, the repair can be formulated to guarantee that fairness constraints hold over $\widehat{\mathcal{H}}^i$:
\begin{equation}
\label{eq:symbolic-repair}
\begin{aligned}
&\min_{\Delta \mathbf{W}, \Delta \mathrm{b}}  \|\Delta \mathbf{W}\|_1 + \|\Delta \mathrm{b}\|_1\ \\ 
\text{s.t. }&\min_{\mathbf{h} \in \widehat{\mathcal{H}}^i} \tilde{f}_{\mathrm{L}}(\mathbf{h}) \geq 0 \vee \max_{\mathbf{h} \in \widehat{\mathcal{H}}^i} \tilde{f}_{\mathrm{L}}(\mathbf{h}) < 0 \\
\end{aligned}
\end{equation}
However, while symbolic bounds provide tighter approximations, they consist of a set of general linear constraints on $\bm{x}$ and $\mathbf{h}$, forming a polytope with fundamentally different geometric structure from hyperrectangles.
Note that one of the core challenge of provable repair stems from the non-linear term $\Delta \mathbf{W} \mathbf{h}$ in fairness constraint $\min_{\mathbf{h}} (\mathbf{W} + \Delta \mathbf{W})\mathbf{h} \geq 0 \vee \max_{\mathbf{h}} (\mathbf{W} + \Delta \mathbf{W})\mathbf{h} < 0  $, where $\mathbf{h} \in [\underline{\mathbf{h}}^i, \overline{\mathbf{h}}^i] \text{ or } \widehat{\mathcal{H}}^i$. 
Although Section~\ref{sec:concrete-repair} handled this issue by exploiting the axis-aligned nature of $[\underline{\mathbf{h}}^i, \overline{\mathbf{h}}^i]$, the polytope $\widehat{\mathcal{H}}^i$ lacks this nature and renders the endpoint analysis in Eq.~\eqref{eq:PQ} inapplicable.

We leverage Lagrangian duality to address this problem. 
In the following, we take \(\min_{\mathbf{h} \in \widehat{\mathcal{H}}^i} \tilde{f}_{\mathrm{L}}(\mathbf{h}) \geq 0\) as an example to illustrate.
The key idea is to recast the minimization over $\mathbf{h}$ as a maximization over dual variables $\bm{\lambda}$, thereby decoupling the multiplication between $\Delta \mathbf{W}$ and $\mathbf{h}$ while preserving soundness. 
Specifically, we first rewrite $ \min _{\mathbf{h} \in \widehat{\mathcal{H}}_{i}} (\mathbf{W}+\Delta \mathbf{W})  \mathbf{h} $ as follows:
\begin{equation}
\label{eq:problem-Cp}
     \min _{\mathbf{p}}\mathbf{C}^{\top} \mathbf{p} \; \;\text{  s.t.  } \mathbf{A}_i  \mathbf{p} \leq \mathbf{D}_i
\end{equation}
where $\mathbf{C}= \left[\mathbf{W} + \Delta \mathbf{W}, \mathbf{0}^{1\times m}\right]^\top
 \in \mathbb{R}^{(m+d)\times1} ,  \mathbf{p} = \begin{pmatrix} \mathbf{h} \\ \bm{x} \end{pmatrix} \in \mathbb{R}^{m+d} $.
$\mathbf{A}_i \in \mathbb{R}^{(2m + 2d) \times (m+d)} $ and $\mathbf{D}_i \in \mathbb{R}^{2m+2d}$ are constant matrices that encode the definition of $\widehat{\mathcal{H}}^i$, constructed as:
\[
\mathbf{A}_i = \begin{bmatrix}
  -\mathbf{I}_{d} & \bm{\alpha}_{i}^{\le} \\
  \mathbf{I}_{d} & -\bm{\alpha}_{i}^{\ge} \\
  \mathbf{0} & -\mathbf{I}_{m} \\
  \mathbf{0} & \mathbf{I}_{m}
\end{bmatrix} , \; \mathbf{D}_i = 
\begin{pmatrix}
-\bm{\beta}_{i}^{\le} \\
\bm{\beta}_{i}^{\ge} \\
-\mathcal{S}_{L}(\bm{x}^i) \\
\mathcal{S}_{U}(\bm{x}^i)
\end{pmatrix}
\]
It is evident that $\mathbf{A}_i  \mathbf{p} \leq \mathbf{D}_i$ is equivalent to the definition of $\widehat{\mathcal{H}}^i$ given in Eq.~\eqref{eq:def-hat-H}.

Recalling Eq.~\eqref{eq:problem-Cp}, we observe that the variable $\Delta \mathbf{W}$ in $\mathbf{C}$ can be temporarily treated as a constant, since the feasible set $\{ {\mathbf{p} \mid \mathbf{A}_i\mathbf{p} \leq \mathbf{D}_i} \}$ is independent of $\Delta \mathbf{W}$.
In other words, problem~\eqref{eq:symbolic-repair} can be viewed as first solving the worst case $\mathbf{p}^*$ under a given $\Delta \mathbf{W}$, then optimizing $\Delta \mathbf{W}$ accordingly.
Under this observation, Eq.~\eqref{eq:problem-Cp} becomes a standard linear programming (LP). 
We define the Lagrangian function $L$ and the dual function $g$ to construct the dual problem of this LP:
\begin{equation}
\begin{aligned}
L(\mathbf{p}, \bm{\lambda}_i) &=\mathbf{C}^{\top} \mathbf{p}+ \bm{\lambda}^{\top}_i  (\mathbf{A}_i  \mathbf{p} - \mathbf{D}_i), \; \bm{\lambda}_i \geq0 \\
g(\bm{\lambda}_i) &= \inf_{\mathbf{p}} L(\mathbf{p}, \bm{\lambda}_i) 
\end{aligned}
\end{equation}
where $\bm{\lambda}_i \in \mathbb{R}^{2m+2d}$ are dual variables.
By maximizing the dual function, we obtain the optimality condition \(\frac{\partial L}{\partial \mathbf{p}}= \mathbf{C} + \mathbf{A}_i^{\top}  \bm{\lambda}_i=0 \Rightarrow \mathbf{A}_i^{\top} \bm{\lambda}_i=- \mathbf{C} \) and construct the dual problem:
\begin{equation}
\label{eq:dual-problem}
\begin{array}{l@{\quad \quad \quad}l}
(\text{Primal Problem}) & (\text{Dual Problem}) \\
\begin{aligned}[t]
    \min_{\mathbf{p}} \quad & \mathbf{C}^{\top} \mathbf{p} \\
    \text{s.t.} \quad & \mathbf{A}_i \mathbf{p} \leq \mathbf{D}_i
\end{aligned}  \; \;
& 
\begin{aligned}[t]
    \max_{\bm{\lambda}_i \geq 0} \quad & -\bm{\lambda}_i^{\top} \mathbf{D}_i \\
    \text{s.t.} \quad & \mathbf{A}_i^{\top} \bm{\lambda}_i = -\mathbf{C}
\end{aligned}
\end{array}
\end{equation}
In the dual problem, the non-linear term $\Delta \mathbf{W} \mathbf{h}$ is decoupled, as only $\bm{\lambda}_i$ are variables and both $\mathbf{D}_i, \mathbf{A}_i$ are constants.
We further note that the primal and dual problems attain the same optimal value due to the strong duality~\cite{boyd2004convex} for LP: 
\begin{equation}
\label{eq:strong-dual}
\min_{\mathbf{A}_i \mathbf{p} \leq \mathbf{D}_i} \mathbf{C}^{\top} \mathbf{p} = 
\max_{\substack{\mathbf{A}_i^{\top} \bm{\lambda}_i = -\mathbf{C} \\ \bm{\lambda}_i \geq 0}} 
- \bm{\lambda}_i^{\top} \mathbf{D}_i
\end{equation}
Similar to the previous section, we now define the variable \(\widehat{\mrm{LB}}_{i}\)
to soundly lower
bound the model's output over \(\widehat{\mathcal{H}}_{i}\):
\begin{equation}
\begin{array}{cc}
  \begin{aligned}
\widehat{\mrm{LB}}_{i}
&=\mathrm{b}+\Delta \mathrm{b}+ (-\bm{\lambda}_i^{\top}  \mathbf{D}_i) \\
&\le \mathrm{b}+\Delta \mathrm{b}+  \max_{\bm{\lambda}_i}  -\bm{\lambda}_i^{\top} \mathbf{D}_i \\
&= \mathrm{b}+\Delta \mathrm{b}+ \min_{\mathbf{A}_i \mathbf{p} \leq \mathbf{D}_i} \mathbf{C}^{\top} \mathbf{p}
\end{aligned}   & \text{ s.t. }\; \begin{aligned}
    &\mathbf{A}_i^{\top} \bm{\lambda}_i = -\mathbf{C} \\
    &\bm{\lambda}_i \geq 0
\end{aligned}\\
\end{array}
\end{equation}
\

Note that the upper bound (for $\mathrm{b}+\Delta \mathrm{b}+ \max_{\mathbf{A}_i \mathbf{p} \leq \mathbf{D}_i} \mathbf{C}^{\top} \mathbf{p}$) can similarly be derived using another set of dual variables $\bm{\eta}_i \in \mathbb{R}^{2m+2d}$.
We can now recast the problem~\eqref{eq:symbolic-repair} as following MILP:
\begin{equation}
\fontsize{9.5pt}{11.4pt}\selectfont
\begin{aligned}
&\quad\quad\quad\quad\min_{\substack{\Delta \mathbf{W}, \Delta \mathrm{b}, \mathbf{Z}, \bm{\lambda}, \bm{\eta}}} \|\Delta \mathbf{W}\|_1 + \|\Delta \mathrm{b}\|_1 \\
 &\text{s.t.} \;  \mathbf{Z} \in\{0,1\}^{n} , \ \mathbf{C}= \left[\mathbf{W} + \Delta \mathbf{W}, \mathbf{0}^{1\times m}\right]^\top \\
& \forall i \in [n] \begin{cases}
\mathbf{A}_i=\begin{bmatrix}
-\mathbf{I}_{d} & \bm{\alpha}_{i}^{\le} \\
\mathbf{I}_{d} & -\bm{\alpha}_{i}^{\ge} \\
\bm{0} & -\mathbf{I}_{m} \\
\bm{0} & \mathbf{I}_{m}
\end{bmatrix},\mathbf{D}_i = 
\begin{pmatrix}
-\bm{\beta}_{i}^{\le} \\
\bm{\beta}_{i}^{\ge} \\
-\mathcal{S}_{L}(\bm{x}^i) \\
\mathcal{S}_{U}(\bm{x}^i)
\end{pmatrix} \\
\mathbf{A}_i^{\top}  \bm{\lambda}_i = -\mathbf{C}, \; \bm{\lambda}_i \geq 0 \\
\mathbf{A}_i^{\top}  \bm{\eta}_i = \mathbf{C}, \; \bm{\eta}_i \geq 0 \\
 \widehat{\mrm{LB}}_{i}=\mathrm{b}+\Delta \mathrm{b}+ (-\bm{\lambda}_i^{\top}  \mathbf{D}_i) \\
 \widehat{\mrm{UB}}_{i}=\mathrm{b}+\Delta \mathrm{b}+ (\bm{\eta}_i^{\top}  \mathbf{D}_i) \\
 \widehat{\mrm{LB}}_{i} \geq  - \mrm{M} \cdot (1 - \mrm{Z}_i) \wedge
\widehat{\mrm{UB}}_{i} <  \mrm{M} \cdot \mrm{Z}_i  \\
\end{cases} \\
\end{aligned}
\label{eq:milp}
\end{equation}
In the last row we employ the Big-M method to convert the disjunction operation to conjunction, where \(\mrm{M}\) is a sufficiently large constant.
This MILP formulation eliminates the non-linear term $\Delta \mathbf{W} \mathbf{h}$ while maintains soundness through strong duality.
The provable repair guarantees achieved by solving Problem~\eqref{eq:milp}, along with its superiority over Problem~\eqref{eq:concrete-repair} in terms of objective value, are formally stated as follows:

\begin{theorem}
\label{thm:main}
Let $\mathcal{D}_{\text{r}}=\{\bm{x}^i\}_{i=1}^n$ be a set of inputs, each associated with a similarity neighborhood $\mathcal{S}(\bm{x}^i)$, and let $f=f_{\mathrm{L}}\circ \tilde{f}_{1:\mathrm{L}}$ be a DNN. 
Then any feasible solution $(\Delta \mathbf{W}, \Delta \mathrm{b}, \mathbf{Z}, \bm{\lambda}, \bm{\eta})$ to the problem~\eqref{eq:milp} induces a repaired final layer $\tilde{f}_{\mathrm{L}}$ such that the composite model $\tilde{f}=\tilde{f}_{\mathrm{L}}\circ \tilde{f}_{1:\mathrm{L}}$ provably satisfies individual fairness on all $\mathcal{S}(\bm{x}^i)$, $ i \in [n]$.
Moreover, the optimal solution of problem~\eqref{eq:milp} is guaranteed to be no worse than that of problem~\eqref{eq:concrete-repair} (i.e., it achieves an objective value that is no larger).
\end{theorem}

\begin{proof}
Consider \((\Delta \mathring{\mathbf{W}}, \Delta \mathring{\mathrm{b}}, \mathring{\mathbf{Z}}, \mathring{\bm{\lambda}}, \mathring{\bm{\eta}})\) as a feasible solution to the MILP problem~\eqref{eq:milp}.
We denote $\mathring{\mathbf{C}}=\begin{pmatrix}
  \mathbf{W} + \Delta \mathring{\mathbf{W}} \\
  \mathbf{0}^m
 \end{pmatrix}$.
Recalling the dual problem constructed in Eq.~\eqref{eq:dual-problem} and the strong duality~\cite{boyd2004convex}, we have:
\begin{equation*}
\begin{aligned}
 \max_{\substack{\mathbf{A}_i^{\top} \bm{\lambda}_i = -\mathring{\mathbf{C}} \\ \bm{\lambda}_i \geq 0}}  -\bm{\lambda}_i^{\top} \mathbf{D}_i
& = \min_{\mathbf{A}_i \mathbf{p} \leq \mathbf{D}_i} \mathring{\mathbf{C}}^{\top} \mathbf{p}
=  \min_{\mathbf{h} \in \widehat{\mathcal{H}}_{i} } (\mathbf{W}+\Delta \mathring{\mathbf{W}}) \mathbf{h}\\
 \min_{\substack{\mathbf{A}_i^{\top} \bm{\eta}_i = \mathring{\mathbf{C}} \\ \bm{\eta}_i \geq 0}}
\bm{\eta}_i^{\top}  \mathbf{D}_i
& =\max_{\mathbf{A}_i \mathbf{p} \leq \mathbf{D}_i} \mathring{\mathbf{C}}^{\top} \mathbf{p} =\max_ {\mathbf{h} \in \widehat{\mathcal{H}}_i}
(\mathbf{W}+\Delta \mathring{\mathbf{W}}) \mathbf{h}    \\
\end{aligned}
\end{equation*}
Thus, the model outputs over \(\mathcal{S}(\bm{x}^i)\) are bounded by $\widehat{\mathrm{LB}}_{i}$ and $\widehat{\mathrm{UB}}_{i}$ as:
\begin{equation*}
\begin{aligned}
&\begin{array}{cc}
\begin{aligned}
 \widehat{\mathrm{LB}}_{i} & = \mathrm{b}+\Delta \mathring{\mathrm{b}}+ (-\mathring{\bm{\lambda}}_i^{\top}  \mathbf{D}_i) \\ & \le \mathrm{b}+\Delta \mathring{\mathrm{b}}+  \max_{\bm{\lambda}_i} -\bm{\lambda}_i^{\top} \mathbf{D}_i\\
&= \mathrm{b}+\Delta \mathring{\mathrm{b}}+ \min_{\substack{\mathbf{h} \in \mathcal{\widehat{H}}_i}}(\mathbf{W}+\Delta \mathring{\mathbf{W}})  \mathbf{h} \\
&= \min_{\substack{\mathbf{h} \in \mathcal{\widehat{H}}_i}}  \tilde{f}_{\mathrm{L}} (\mathbf{h}) \\
&\textcolor{blue}{\le} \min_{\substack{\mathbf{h} \in \mathcal{H}_i}}  \tilde{f}_{\mathrm{L}} (\mathbf{h})= \min_{\substack{\bm{x} \in \mathcal{S}(\bm{x}^i)}}  \tilde{f} (\bm{x}) \\
 \end{aligned}   & \text{ s.t. }\; \begin{aligned}
    &\mathbf{A}_i^{\top} \bm{\lambda}_i = -\mathbf{C} \\
    &\bm{\lambda}_i \geq 0
\end{aligned}\\
\end{array}
\\
&\begin{array}{cc}
\begin{aligned}
\widehat{\mathrm{UB}}_{i} &=\mathrm{b}+\Delta \mathring{\mathrm{b}}+ (\bm{\mathring{\eta}}_i^{\top}  \mathbf{D}_i) \\
&\ge \mathrm{b}+\Delta \mathring{\mathrm{b}}+  \min_ {\bm{\eta}_i} \bm{\eta}_i^{\top}  \mathbf{D}_i  \\
&= \mathrm{b}+\Delta \mathring{\mathrm{b}}+  \max_{\substack{\mathbf{h} \in \mathcal{\widehat{H}}_i}}(\mathbf{W}+\Delta \mathring{\mathbf{W}}) \mathbf{h} \\
&=\max_{\substack{\mathbf{h} \in \mathcal{\widehat{H}}_i}}  \tilde{f}_{\mathrm{L}} (\mathbf{h}) \\
&\textcolor{blue}{\ge} \max_{\substack{\mathbf{h} \in \mathcal{H}_i}}  \tilde{f}_{\mathrm{L}} (\mathbf{h}) = \max_{\substack{\bm{x} \in \mathcal{S}(\bm{x}^i)}}  \tilde{f} (\bm{x})
\end{aligned}
 & \text{ s.t. }\; \begin{aligned}
    &\mathbf{A}_i^{\top} \bm{\eta}_i = \mathbf{C} \\
    &\bm{\eta}_i \geq 0
\end{aligned}\\
\end{array}
\end{aligned}
\end{equation*}
where the \textcolor{blue}{blue} inequality holds due to the soundness of the symbolic bounds, i.e., $\mathcal{H}_i \subseteq \mathcal{\widehat{H}}_i $.
Finally, by integrating \( \widehat{\mathrm{LB}}_{i} \geq  - \mathrm{M} \cdot (1 - \mathring{\mathrm{Z}}_i) \wedge
\widehat{\mathrm{UB}}_{i} <  \mathrm{M} \cdot \mathring{\mathrm{Z}}_i\), the model's outputs over \(\mathcal{S}(\bm{x}^i) \) are strictly confined to either positive or negative regimes, guaranteeing consistent classifications and thereby ensuring individual fairness.

We next prove that the optimal solution of problem~\eqref{eq:milp} is guaranteed to be no worse than that of problem~\eqref{eq:concrete-repair}.
The key point is that any feasible solution $(\Delta \dot{\mathbf{W}}, \Delta \dot{\mathrm{b}})$ of the problem~\eqref{eq:concrete-repair} can be extended to a feasible solution $(\Delta \mathbf{W}, \Delta \mathrm{b}, \mathbf{Z}, \bm{\lambda}, \bm{\eta})$ of the MILP problem~\eqref{eq:milp}.
This is because the symbolic bounds encoded in problem~\eqref{eq:milp} are tighter than the concrete bounds, which implies $\mathcal{\widehat{H}}_i \subseteq [\underline{\mathbf{h}}^i, \overline{\mathbf{h}}^i]$, thus we have:
\begin{equation*}
\begin{aligned}
     \min_{\mathbf{h} \in [\underline{\mathbf{h}}^i, \overline{\mathbf{h}}^i]} (\mathbf{W} + \Delta \dot{\mathbf{W}}) \mathbf{h}
    \le  \min_{\substack{\mathbf{h} \in \mathcal{\widehat{H}}_i}}(\mathbf{W}+\Delta \dot{\mathbf{W}})  \mathbf{h} \\
 \max_{\mathbf{h} \in [\underline{\mathbf{h}}^i, \overline{\mathbf{h}}^i]} (\mathbf{W} + \Delta \dot{\mathbf{W}}) \mathbf{h}
    \ge  \max_{\substack{\mathbf{h} \in \mathcal{\widehat{H}}_i}}(\mathbf{W}+\Delta \dot{\mathbf{W}})  \mathbf{h} \\
\end{aligned}
\end{equation*}
In other words, the fairness constraints in problem~\eqref{eq:milp} are relaxed compared to those in problem~\eqref{eq:concrete-repair}.
Therefore, the feasible region of problem~\eqref{eq:concrete-repair} is a subset of that of problem~\eqref{eq:milp}.
Consequently, problem~\eqref{eq:milp} minimizes the same objective over a (potentially) larger feasible region, and hence its optimal value is no greater than that of problem~\eqref{eq:concrete-repair}.
\qedhere
\end{proof}

The overall algorithm of \name is shown in Algorithm~\ref{algorithm:provable-repair}.
It first calibrates the feature extractor $f_{1:\mathrm{L}}$, then synthesizes symbolic bounds to formulate an MILP problem, which is solved by Gurobi~\cite{gurobi} to achieve provable repair.

\noindent \textbf{Remark:} While our illustration focused on binary classification, we can naturally extend \name to multi-class tasks by adapting the second step. 
For each input $\bm{x}^{i}$, instead of computing scalar lower/upper bounds for a single output neuron, we compute vectorized bounds ($\widehat{\mathrm{LB}}_i,\widehat{\mathrm{UB}}_i\in \mathbb{R}^{N}$) for all $N$ classes. 
We then vectorize the integer variables $\mathrm{Z}_i$ in a one-hot manner to represent the multi-class decision,  i.e., $\mathrm{Z}_i \in \{0,1\}^{N}, \sum_j \mathrm{Z}_{i,j}=1$.
These modifications enable us to soundly formulate the MILP for multi-class task.

\section{Evaluation}
In this section, we evaluate \name through extensive experiments and answer the following research questions.
\begin{itemize}
\item \textbf{RQ1:} How effective is \name in repairing NN unfairness?
\item \textbf{RQ2:} Can \name handle multiple protected attributes and more practical fairness definition?
\item \textbf{RQ3:} How effective is \name compared to baselines when evaluated by existing fairness testing frameworks?
\item \textbf{RQ4:} How do the two key steps of \name  contribute to the repair performance?
\end{itemize}

\subsection{Experimental Setup}

\subsubsection{Datasets and Models} We conduct experiments on four datasets including Adult (Census Income)~\cite{misc_adult_2}, German Credit~\cite{misc_statlog_(german_credit_data)_144}, Bank Marketing~\cite{misc_bank_marketing_222} and Compas~\cite{compas2016propublica}, 
which are commonly used in fairness testing~\cite{zhang2020white, zhang2021efficient, wang2024maft, zheng2022neuronfair} and repair~\cite{sun2022causality, quan2025dissecting}.
\ding{182} Adult is a dataset to predict whether a person earns more than \$50\,000. The protected attributes are gender, race, and age. \ding{183} Bank is a dataset for predicting if the bank client will subscribe a term deposit, with age being the sensitive attribute. \ding{184} German Credit is a dataset to assess creditworthiness based on personal and financial records. It has two protected attributes: gender and age. \ding{185} Compas is a dataset for assessing the likelihood of a criminal defendant reoffending. The protected attributes are gender, race, and age.

The biased NNs for repair are sourced from previous fairness verification works Fairify~\cite{biswas2023fairify} and FairQuant~\cite{kim2024fairquant}. These models consist of various number of layers and neurons.

\begin{algorithm}[t]
\caption{\name }
\label{algorithm:provable-repair}
\DontPrintSemicolon
\KwIn {Biased NN $f=f_{\mathrm{L}} \circ f_{1:\mathrm{L}}$, repair set \(\mathcal{D}_{\text{r}}\), a small set of training data \(\mathcal{D}_{\text{c}}\), maximum iteration $\mathrm{T}$}
\KwOut {Repaired model $\tilde{f}$}

$\tilde{f}_{1:\mathrm{L}} \gets \textsc{BiasedFeatureCalibration}(f, \mathcal{D}_{\text{r}}, \mathcal{D}_{\text{c}}, \mathrm{T})$

\mytcc{\small Repair $f_{1:\mathrm{L}}$ by Alg.~1}

\For{$i \gets 1$ \KwTo $|\mathcal{D}_{\text{r}}|$} {
        $\mathcal{S}(\bm{x}^i) \eqdef $ Set of all inputs similar to $\bm{x}^i$
        
        $\mathbf{\mathcal{\widehat{H}}}_i \gets \textsc{SymbolicBound}(\tilde{f}_{1:\mathrm{L}}, \mathcal{S}(\bm{x}^i))$

    }

Construct the MILP problem $\mathcal{M}$ according to Eq.~\eqref{eq:milp}

$\Delta \mathbf{W}, \Delta \mathrm{b}, \mathbf{Z}, \bm{\lambda}, \bm{\eta} \gets \textrm{Solve}(\mathcal{M})$

$\tilde{f}_{\mathrm{L}} \gets \textrm{Linear($\mathbf{W}+\Delta \mathbf{W}, b+\Delta \mathrm{b}$)}$


\Return $\tilde{f}_{\mathrm{L}} \circ \tilde{f}_{1:\mathrm{L}} $

\end{algorithm}

\subsubsection{Repair Setup}
The original datasets $\mathcal{D}_{\text{full}}$ are divided into the training set $\mathcal{D}_{\text{train}}$ and the test set $\mathcal{D}_{\text{test}}$. 
Then we randomly select 100 instances from $\mathcal{D}_{\text{train}}$ to construct the repair set $\mathcal{D}_{\text{r}}$.
We also provide a small set $\mathcal{D}_{\text{c}} \subset \mathcal{D}_{\text{train}}$ for all methods to preserve model's performance.
Table~\ref{tab:exp-config} shows the size of these sets and the original accuracy of the DNN for each dataset.

\begin{table}[h]
\caption{Details of datasets involved in repair and evaluation.}
\small
\renewcommand{\arraystretch}{0.9}
\setlength{\tabcolsep}{5.5pt}
\label{tab:exp-config}
\begin{tabular}{c|c|ccc|c|c}
\toprule
Dataset   & $\mathcal{D}_{\text{full}}$ & $\mathcal{D}_{\text{train}}$ &  $\mathcal{D}_{\text{r}} $ & $\mathcal{D}_{\text{c}}$ & $\mathcal{D}_{\text{test}}$ & Original Acc \\
\midrule
Adult & 48\,842 & 38\,438  & 100 & 500 &6\,784 & 85.24\% \\
Compas & 6\,172 & 4\,937 & 100 & 500 & 1\,235 & 73.85\%\\
German & 1\,000 & 700 & 100 & 100 & 300 & 72.67\%\\
Bank & 45\,211 & 36\,168  & 100 & 500 & 9\,043 & 89.53\%\\
\bottomrule
\end{tabular}
\end{table}

\begin{table*}[t]
\caption{Fairness improvement results for single-attribute repair. The best values are highlighted in \textbf{bold}.}
\label{table:exp-single}
\centering
\renewcommand{\arraystretch}{0.96}
\footnotesize
\begin{tabular}{cc|ccc|ccc|cc|c|l}
\toprule
\multirow{2}{*}{Metrics} & \multirow{2}{*}{Methods} & \multicolumn{3}{c|}{Compas} & \multicolumn{3}{c|}{Adult} & \multicolumn{2}{c|}{German} & Bank & \multirow{2}{*}{Avg (Relative Drop)} \\
 & & gender & race & age & gender & race & age & gender & age & age & \\
\midrule
 \multirow{5}{*}{CUR} & Original & 5.00\% & 9.00\% & 42.00\% & 3.00\% & 2.00\% & 14.00\% & 3.00\% & 4.00\% & 1.00\% & 9.22\% (-)\\
& {FLIP} & \textbf{0.00\%} & 6.00\% & 2.00\% & 2.00\% & 4.00\% & 8.00\% & 7.00\% & 1.00\% & 2.00\% & {3.56}\% (61.45\%$\downarrow$) \\
& {CARE} & \textbf{0.00\%} & 2.00\% & 17.00\% & 2.00\% & 2.00\% & 6.00\% & \textbf{0.00\%} & 1.00\% & \textbf{0.00\%} & {3.33}\% (63.86\%$\downarrow$) \\
& {GRFT} & \textbf{0.00\%} & 1.00\% & 16.00\% & \textbf{0.00\%} & 1.00\% & 7.00\% & \textbf{0.00\%} & \textbf{0.00\%} & \textbf{0.00\%} & {2.78}\% (69.88\%$\downarrow$) \\
& {Ours} & \textbf{0.00\%} & \textbf{0.00\%} & \textbf{0.00\%} & \textbf{0.00\%} & \textbf{0.00\%} & \textbf{0.00\%} & \textbf{0.00\%} & \textbf{0.00\%} & \textbf{0.00\%} & \textbf{0.00\%} \textbf{(100.00\%$\downarrow$)} \\
\cmidrule{1-12}
 \multirow{5}{*}{IDI-D} & Original & 6.25\% & 8.77\% & 44.86\% & 2.74\% & 4.15\% & 21.64\% & 2.60\% & 2.90\% & 1.12\% & 10.56\% (-)\\
& {FLIP} & 1.80\% & 8.25\% & 4.83\% & 2.95\% & 5.56\% & 8.14\% & 4.30\% & 0.40\% & 0.93\% & {4.13}\% (60.91\%$\downarrow$) \\
& {CARE} & 3.19\% & 5.75\% & 16.96\% & 1.89\% & 2.69\% & 5.60\% & 0.10\% & 3.30\% & 1.39\% & {4.54}\% (56.99\%$\downarrow$) \\
& {GRFT} & \textbf{0.00\%} & 2.37\% & 19.91\% & 0.98\% & \textbf{0.64\%} & 2.38\% & 0.20\% & 0.10\% & 0.56\% & {3.02}\% (71.44\%$\downarrow$) \\
& {Ours} & \textbf{0.00\%} & \textbf{0.00\%} & \textbf{0.49\%} & \textbf{0.79\%} & 1.48\% & \textbf{0.83\%} & \textbf{0.00\%} & \textbf{0.00\%} & \textbf{0.29\%} & \textbf{0.43\%} \textbf{(95.93\%$\downarrow$)} \\
\cmidrule{1-12}
 \multirow{5}{*}{IDI-S} & Original & 1.76\% & 2.73\% & 10.94\% & 0.65\% & 0.53\% & 3.65\% & 7.43\% & 8.77\% & 6.67\% & 4.79\% (-)\\
& {FLIP} & 0.63\% & 1.61\% & 3.08\% & 0.43\% & 0.33\% & 1.25\% & 2.99\% & 3.76\% & 3.06\% & {1.90}\% (60.27\%$\downarrow$) \\
& {CARE} & 0.75\% & 0.93\% & 11.36\% & 0.75\% & \textbf{0.28\%} & 3.30\% & 2.46\% & 7.14\% & 4.40\% & {3.49}\% (27.24\%$\downarrow$) \\
& {GRFT} & \textbf{0.00\%} & 0.50\% & 1.74\% & 0.59\% & 0.48\% & 2.51\% & 1.98\% & 2.31\% & 1.67\% & {1.31}\% (72.67\%$\downarrow$) \\
& {Ours} & \textbf{0.00\%} & \textbf{0.00\%} & \textbf{0.96\%} & \textbf{0.12\%} & 0.30\% & \textbf{0.26\%} & \textbf{0.00\%} & \textbf{0.00\%} & \textbf{1.31\%} & \textbf{0.33\%} \textbf{(93.16\%$\downarrow$)} \\
\bottomrule
\end{tabular}
\end{table*}

\subsubsection{Baselines}
We select the following NN individual fairness repair methods as baselines for comprehensive comparison:
\ding{182} \textbf{Flipping-based fine-tuning}. This is a basic method for flipping the protected attributes (e.g., changing the gender from female to male) in each input while maintaining the ground truth labels.
Through this learning paradigm, the biased model learns to make predictions that are insensitive to variations in protected attributes.
\ding{183} \textbf{CARE}~\cite{sun2022causality}. This method utilizes the causal model to pinpoint key neurons that are responsible for unfairness, followed by the PSO algorithm to generate neuron-level patches for repairing the model.
\ding{184} \textbf{GRFT}~\cite{quan2025dissecting}. As the latest state-of-the-art individual fairness repair framework, GRFT designs a loss function aimed at directly reducing differences in model outputs between IDI pairs.

Our method involves two hyperparameters: the number of iterations $\mathrm{T}$ in Step 1, and the learning rate. We fix them to 200 and 0.001 across all scenarios.

\subsubsection{Evaluation Metrics}
We evaluate the repaired models from two perspectives: improvement in fairness and preservation of original performance. We also record the time costs of all methods. Details of fairness metrics are described below:

\noindent \textbf{Fairness improvement}. 
Three metrics are considered for the fairness evaluation.
The Certified Unfair Rate (CUR) assesses the percentage of inputs $\bm{x}$ in the repair set $\mathcal{D}_{\text{r}}$ for which the repaired model $\tilde{f}$ still produces biased outputs over $\mathcal{S}(\bm{x})$:
\[
\text{CUR} = \frac{1}{|\mathcal{D}_{\text{r}}|} \sum_{i=1}^{|\mathcal{D}_{\text{r}}|}\mathbb{I}\left(\mathcal{U}(\tilde{f}, \bm{x}^i, \mathcal{S}(\bm{x}^i)) \neq \emptyset \right)
\]
where $\mathbb{I}$ is the indicator function and $\mathcal{U}$ is the set of all IDI instances in $\mathcal{S}(\bm{x})$ as defined in Eq.~\eqref{eq:U-set}.
For finite $\mathcal{S}(\bm{x}^i)$ (e.g., the protected attribute is gender and all non-sensitive attributes are fixed), we enumerate all possible items to check whether $\mathcal{U}(\tilde{f}, \bm{x}^i, \mathcal{S}(\bm{x}^i))$ is empty. 
For infinite cases, we employ the method from~\cite{DBLP:conf/iclr/TjengXT19} that certify the model by solving an MILP. 
This MILP precisely computes the model output ranges so that we can verify whether $\tilde{f}$ is fair over $\mathcal{S}(\bm{x}^i)$.
\emph{A provable repair must guarantee that $\tilde{f}$ achieves zero CUR.}

An effective repair should generalize to unseen inputs. 
We assess this using two metrics: 
IDI-D, the proportion of inputs in the entire dataset $\mathcal{D}_{\text{full}}$ that are IDIs; 
and IDI-S, the proportion of inputs identified as IDIs in a set $\mathcal{D}_{\text{sample}}$ consisting of 100\,000 instances randomly sampled from full input space, following prior work~\cite{zhang2020white, zhang2021efficient, zheng2022neuronfair}.
To compute IDI-D and IDI-S, we enumerate all similar inputs in $\mathcal{S}(\bm{x})$ for each instance $\bm{x} \in \mathcal{D}_{\text{full}} \text{ or } \mathcal{D}_{\text{sample}}$ when $\mathcal{S}(\bm{x})$ is finite; otherwise, we conduct uniform sampling over $\mathcal{S}(\bm{x})$ to obtain statistically significant estimates.
For both metrics, \emph{lower} values indicate better fairness improvement after repair.

In addition to the above metrics, we further incorporate four state-of-the-art fairness testing frameworks to evaluate \name and all baselines: ADF~\cite{zhang2020white}, EIDIG~\cite{zhang2021efficient}, NeuronFair~\cite{zheng2022neuronfair}, and GRFT~\cite{quan2025dissecting}.
We configured these frameworks following their technical papers and used them to generate IDIs for both the original and the repaired models. 
Fewer IDIs detected on the repaired models indicate better repair effectiveness.




\subsection{RQ1: How effective is \name in repairing NN unfairness?}
To answer this research question, we evaluate the performance of all repair methods. The results are summarized in Table~\ref{table:exp-single}.
As we can see, all methods lead to a notable reduction in the Certified Unfairness Rate (CUR), while \name consistently achieves provable fairness guarantees, reducing CUR to 0 across all settings. In contrast, existing baselines provide no theoretical guarantees and only reduce CUR by approximately 60--70\%.
In terms of generalization to unseen inputs, our method also outperforms all baselines. 
On both the full dataset $\mathcal{D}_{\text{full}}$ and the sampled set $\mathcal{D}_{\text{sample}}$, the models repaired by \name yield extremely low unfairness: the average IDI-D and IDI-S rates are reduced to just 0.43\% and 0.33\%, corresponding to relative reductions of 95.93\% and 93.16\%, respectively. 
The best-performing baseline, GRFT, only achieves 71.44\% and 72.67\% reductions in IDI-D and IDI-S.

\begin{table}[t]
\caption{Accuracy of repaired models, where colors denote accuracy loss: green ($\le$3\%), orange (3–5\%), red ($>$5\%).}
\label{table:exp-acc-single}
\small
\footnotesize
\centering
\begin{tabular}{cc|cccc}
\toprule
Dataset & Prot. Attr. &  FLIP & CARE & GRFT & Ours \\ 
\midrule
 Compas & gender & \cellcolor{green!20}72.39\% & \cellcolor{green!20}73.44\% & \cellcolor{green!20}72.71\% & \cellcolor{green!20}72.06\% \\
 Compas & race & \cellcolor{green!20}72.87\% & \cellcolor{green!20}71.42\% & \cellcolor{green!20}73.60\% & \cellcolor{green!20}72.47\% \\
 Compas & age & \cellcolor{orange!20}70.53\% & \cellcolor{red!20}65.43\% & \cellcolor{red!20}63.24\% & \cellcolor{orange!20}70.28\% \\
 Adult & gender & \cellcolor{green!20}83.65\% & \cellcolor{green!20}84.39\% & \cellcolor{orange!20}82.06\% & \cellcolor{green!20}82.50\% \\
 Adult & race & \cellcolor{green!20}83.58\% & \cellcolor{green!20}84.52\% & \cellcolor{orange!20}81.07\% & \cellcolor{orange!20}82.05\% \\
 Adult & age & \cellcolor{green!20}84.51\% & \cellcolor{orange!20}80.78\% & \cellcolor{orange!20}80.26\% & \cellcolor{green!20}83.67\% \\
 German & gender & \cellcolor{orange!20}68.67\% & \cellcolor{orange!20}69.33\% & \cellcolor{green!20}70.00\% & \cellcolor{green!20}69.67\% \\
 German & age & \cellcolor{green!20}70.33\% & \cellcolor{green!20}72.67\% & \cellcolor{green!20}70.00\% & \cellcolor{green!20}69.67\% \\
 Bank & age & \cellcolor{green!20}89.28\% & \cellcolor{green!20}88.80\% & \cellcolor{green!20}89.33\% & \cellcolor{green!20}89.35\% \\
\midrule  \multicolumn{2}{c|}{Average} & \cellcolor{green!20}77.31\% & \cellcolor{green!20}76.75\% & \cellcolor{orange!20}75.81\% & \cellcolor{green!20}76.86\% \\
 \multicolumn{2}{c|}{Accuracy Change} & \cellcolor{green!20}-\,1.81\% & \cellcolor{green!20}-\,2.37\% & \cellcolor{orange!20}-\,3.32\% & \cellcolor{green!20}-\,2.27\% \\
\bottomrule
\end{tabular}
\vspace{-1em}
\end{table}

Regarding performance preservation, we present the accuracy of the repaired models in Table~\ref{table:exp-acc-single}.
We find that all methods introduce some accuracy degradation. 
Among them, FLIP and \name exhibit more stable performance, with average accuracy losses of 1.81\% and 2.27\%, respectively. 
By comparison, CARE and GRFT can have a notable negative impact on model performance in certain scenarios. 
For instance, when repairing the Age attribute on the Compas dataset, they result in accuracy drops of nearly 10\%.

As for efficiency, both FLIP and GRFT introduce negligible time overheads, taking around 1 second per setting. 
\name also maintains a desirable runtime of 18 seconds on average. 
On the other hand, CARE incurs more time consumption (averaging 244 seconds) due to its reliance on the PSO algorithm, which often demands extensive iteration for convergence.

\begin{table*}[t]
\caption{Results of fairness improvements with multiple attributes and relaxed fairness properties (G, R and A denote Gender, Race and Age). Details of all relaxed fairness specifications are provided at the end of this section.}
\label{table:exp-mutli-eps}
\setlength{\tabcolsep}{4.5pt}
\centering
\footnotesize
\begin{tabular}{cc|llll|llll|llll}
\toprule
\multirow{2}{*}{Dataset} & \multirow{2}{*}{Attr.} & \multicolumn{4}{c|}{CUR} & \multicolumn{4}{c|}{IDI-D} &  \multicolumn{4}{c}{IDI-S} \\
& & Original & FLIP & GRFT & Ours & Original & FLIP & GRFT & Ours& Original & FLIP & GRFT & Ours \\
\midrule
Compas & G-R& 19.00\% & 7.00\% & 3.00\% & \textbf{0.00\%} & 16.92\% & 7.63\% & 6.59\% & \textbf{0.00\%} & 4.49\% & 1.94\% & 1.30\% & \textbf{0.00\%}  \\
Compas & G-A& 49.00\% & 13.00\% & 20.00\% & \textbf{0.00\%} & 51.65\% & 7.57\% & 25.84\% & \textbf{1.62\%} & 12.71\% & 3.57\% & 1.57\% & \textbf{1.33\%}  \\
Compas & R-A& 55.00\% & 13.00\% & 13.00\% & \textbf{0.00\%} & 55.54\% & 7.18\% & 15.04\% & \textbf{3.68\%} & 13.56\% & 2.91\% & \textbf{1.91\%} & 2.56\%  \\
Adult & G-R& 7.00\% & 1.00\% & 2.00\% & \textbf{0.00\%} & 7.59\% & 6.80\% & 3.21\% & \textbf{1.78\%} & 1.18\% & 0.73\% & 1.11\% & \textbf{0.39\%}  \\
Adult & G-A& 16.00\% & 6.00\% & 6.00\% & \textbf{0.00\%} & 23.01\% & 10.51\% & 7.60\% & \textbf{3.13\%} & 4.31\% & 1.16\% & 3.74\% & \textbf{0.78\%}  \\
Adult & R-A& 17.00\% & 8.00\% & 7.00\% & \textbf{0.00\%} & 23.01\% & 12.62\% & 7.51\% & \textbf{2.71\%} & 4.20\% & 1.36\% & 3.52\% & \textbf{0.78\%}  \\
German & G-A& 4.00\% & 13.00\% & \textbf{0.00\%} & \textbf{0.00\%} & 6.10\% & 8.90\% & 0.40\% & \textbf{0.00\%} & 16.11\% & 9.91\% & 4.39\% & \textbf{0.00\%}  \\
 \midrule 
 \multicolumn{2}{c|}{Average}    & 23.86\%   & 8.71\%   & 7.29\%   & \textbf{0.00\%}  & 26.26\%   & 8.74\%   & 9.46\%   & \textbf{1.85\%}  & 8.08\%   & 3.08\%   & 2.51\%   & \textbf{0.83\%}\\
 \multicolumn{2}{c|}{Relative Drop}  &  -  & 63.47\%$\downarrow$  & 69.46\%$\downarrow$  & \textbf{100.00\%}$\downarrow$&  -  & 66.70\%$\downarrow$  & 63.99\%$\downarrow$  & \textbf{92.97\%}$\downarrow$&  -  & 61.88\%$\downarrow$  & 69.00\%$\downarrow$  & \textbf{89.67\%}$\downarrow$\\
\midrule
\midrule
Adult & G-$\phi_{A}$& 3.00\% & 6.00\% & 1.00\% & \textbf{0.00\%} & 3.26\% & 4.97\% & 1.43\% & \textbf{1.22\%} & 0.73\% & 0.52\% & 0.67\% & \textbf{0.18\%}  \\
Adult & A-$\phi_{A}$& 14.00\% & 6.00\% & 5.00\% & \textbf{0.00\%} & 21.98\% & 6.67\% & 6.68\% & \textbf{0.61\%} & 3.73\% & 0.90\% & 3.07\% & \textbf{0.32\%}  \\
Adult & R-$\phi_{A}$& 2.00\% & 4.00\% & 1.00\% & \textbf{0.00\%} & 4.70\% & 6.42\% & 1.78\% & \textbf{1.20\%} & 0.62\% & 0.39\% & 0.60\% & \textbf{0.32\%}  \\
Bank & A-$\phi_{B}$& 2.00\% & 5.00\% & 2.00\% & \textbf{0.00\%} & 3.56\% & 3.72\% & 2.04\% & \textbf{0.11\%} & 23.75\% & 18.80\% & 6.82\% & \textbf{2.75\%}  \\
German & G-$\phi_{G}$& 3.00\% & 8.00\% & \textbf{0.00\%} & \textbf{0.00\%} & 3.10\% & 5.30\% & 0.20\% & \textbf{0.00\%} & 7.98\% & 4.01\% & 2.07\% & \textbf{0.00\%}  \\
German & A-$\phi_{G}$& 4.00\% & \textbf{0.00\%} & \textbf{0.00\%} & \textbf{0.00\%} & 3.50\% & 0.10\% & 0.10\% & \textbf{0.00\%} & 9.35\% & 1.99\% & 2.39\% & \textbf{0.00\%}  \\
 \midrule 
 \multicolumn{2}{c|}{Average}    & 4.67\%   & 4.83\%   & 1.50\%   & \textbf{0.00\%}  & 6.68\%   & 4.53\%   & 2.04\%   & \textbf{0.52\%}  & 7.69\%   & 4.44\%   & 2.60\%   & \textbf{0.59\%}\\
 \multicolumn{2}{c|}{Relative Drop}  &  -  & -3.57\%$\downarrow$  & 67.86\%$\downarrow$  & \textbf{100.00\%}$\downarrow$&  -  & 32.22\%$\downarrow$  & 69.49\%$\downarrow$  & \textbf{92.17\%}$\downarrow$&  -  & 42.34\%$\downarrow$  & 66.18\%$\downarrow$  & \textbf{92.27\%}$\downarrow$\\
\bottomrule
\end{tabular}
\vspace{-1em}
\end{table*}

\subsection{RQ2: Can \name handle multiple protected attributes and more practical fairness definition?}
To answer this question, we conduct experiments on repair with multiple sensitive attributes simultaneously. 
CARE is excluded from this evaluation since its current implementation does not support this setting.

As shown in the upper panel of Table~\ref{table:exp-mutli-eps}, the original models exhibit higher unfairness prior to repair. 
This is expected since multiple protected attributes can combine to form a larger set of similar individuals, making it naturally more challenging for the model to satisfy fairness constraints.
Despite this increased challenge, \name consistently achieves provable repair by completely eliminating unfair behaviors on the repair set (i.e., CUR = 0).
Models repaired by FLIP and GRFT, however, often continue to exhibit discrimination, with CUR reductions under 70\%.
Our method also demonstrates remarkable generalization compared to existing approaches. 
It attains significantly lower IDI-D and IDI-S scores across all settings, with average mitigations of 92.97\% and 89.67\%, respectively, except for the ``Compas + Race-Age'' scenario, where its performance is only 0.6\% below the best method.

To further examine the effectiveness of all methods under more practical fairness definitions, we conduct a set of experiments on repairing models with relaxed fairness properties.
Following Fairify~\cite{biswas2023fairify}, we define small perturbations on one non-sensitive attribute so that two individuals are considered similar even if they are not exactly equal on this attribute. 
These individuals are still expected to be assigned to the same class.
For instance, the specification $\phi_{A}$ indicates that for the Adult dataset, individuals are considered similar as long as the value of ``hours-per-week'' attribute differ by at most 1, regardless of their values on sensitive attributes.
The complete list of relaxed fairness specifications is as follows:
\begin{itemize}[left=0pt]
    \item $\phi_{A}$: Two individuals are considered similar if their ``hours-per-week'' attribute differs by at most 1.
    \item $\phi_{B}$: Two individuals are regarded as similar when their ``duration'' attribute varies by no more than 1.
    \item $\phi_{G}$: This specification defines similarity between individuals whose ``credit-amount'' attribute lies within a fixed interval of size 100.
\end{itemize}

The results are shown in the lower panel of Table~\ref{table:exp-mutli-eps}.
We observe that FLIP suffers a significant performance drop compared to the previous repair setting; in some cases, the unfairness of the repaired model even increases.
This may be because it does not account for similarity between individuals whose non-sensitive attributes differ slightly and thus fails to enforce fairness under relaxed definitions.
GRFT maintains similar performance as before, achieving around 70\% reduction in unfairness metrics.
In contrast, \name consistently outperforms existing approaches across all scenarios, achieving complete unfairness elimination on the repair set, along with over 90\% reduction in both IDI-D and IDI-S.

In terms of accuracy, all methods exhibit slightly increased performance degradation compared to the single-attribute repair setting. Both FLIP and \name maintain stable performance, with average accuracy losses of 2.12\% and 2.90\%, respectively, which are significantly better than GRFT (5.03\%).

\begin{table*}[t]
\caption{Results of fairness evaluation with various testing frameworks.
Colors indicate repair effectiveness: green for strong unfairness mitigation (ratios near 0), orange for minor change (ratios near 1), and red for increase (ratios above 1).}
\label{table:testing}
\setlength{\tabcolsep}{5.1pt}
\renewcommand{\arraystretch}{0.95}
\footnotesize
\begin{tabular}{lc|cccc|cccc|cccc|cccc}
\toprule
\multirow{2}{*}{Dataset} & \multirow{2}{*}{Attr.} & \multicolumn{4}{c|}{FLIP} & \multicolumn{4}{c|}{CARE} & \multicolumn{4}{c|}{GRFT} & \multicolumn{4}{c}{\name} \\
 &  & ADF & EIDIG & NF & GRFT & ADF & EIDIG & NF & GRFT & ADF & EIDIG & NF & GRFT & ADF & EIDIG & NF & GRFT \\
\midrule
\multirow{6}{*}{Compas} & G & \cellcolor{green!25}0.28 & \cellcolor{green!25}0.24 & \cellcolor{green!25}0.36 & \cellcolor{green!25}0.33 & \cellcolor{green!25}0.28 & \cellcolor{green!25}0.36 & \cellcolor{green!20}0.45 & \cellcolor{green!20}0.42 & \cellcolor{green!30}\textbf{0.00} & \cellcolor{green!30}\textbf{0.00} & \cellcolor{green!30}\textbf{0.00} & \cellcolor{green!30}\textbf{0.00} & \cellcolor{green!30}\textbf{0.00} & \cellcolor{green!30}\textbf{0.00} & \cellcolor{green!30}\textbf{0.00} & \cellcolor{green!30}\textbf{0.00} \\
 & R & \cellcolor{green!20}0.53 & \cellcolor{green!15}0.63 & \cellcolor{green!20}0.59 & \cellcolor{green!25}0.37 & \cellcolor{green!20}0.42 & \cellcolor{green!20}0.45 & \cellcolor{green!25}0.35 & \cellcolor{green!25}0.22 & \cellcolor{green!30}0.17 & \cellcolor{green!30}0.18 & \cellcolor{green!30}0.18 & \cellcolor{green!30}0.11 & \cellcolor{green!30}\textbf{0.00} & \cellcolor{green!30}\textbf{0.00} & \cellcolor{green!30}\textbf{0.00} & \cellcolor{green!30}\textbf{0.00} \\
 & A & \cellcolor{green!25}0.28 & \cellcolor{green!25}0.34 & \cellcolor{green!25}0.28 & \cellcolor{green!25}0.25 & \cellcolor{green!10}0.92 & \cellcolor{green!10}0.95 & \cellcolor{orange!10}1.04 & \cellcolor{green!10}0.91 & \cellcolor{green!25}0.21 & \cellcolor{green!30}0.20 & \cellcolor{green!30}0.16 & \cellcolor{green!30}0.14 & \cellcolor{green!30}\textbf{0.12} & \cellcolor{green!30}\textbf{0.13} & \cellcolor{green!30}\textbf{0.09} & \cellcolor{green!30}\textbf{0.08} \\
 & G-R & \cellcolor{green!25}0.29 & \cellcolor{green!25}0.32 & \cellcolor{green!25}0.29 & \cellcolor{green!30}0.17 & \cellcolor{green!25}0.32 & \cellcolor{green!25}0.34 & \cellcolor{green!25}0.29 & \cellcolor{green!30}0.17 & \cellcolor{green!20}0.44 & \cellcolor{green!20}0.44 & \cellcolor{green!25}0.29 & \cellcolor{green!30}0.17 & \cellcolor{green!30}\textbf{0.00} & \cellcolor{green!30}\textbf{0.00} & \cellcolor{green!30}\textbf{0.00} & \cellcolor{green!30}\textbf{0.00} \\
 & G-A & \cellcolor{green!15}0.70 & \cellcolor{green!15}0.72 & \cellcolor{green!15}0.64 & \cellcolor{green!20}0.51 & \cellcolor{green!15}0.79 & \cellcolor{green!15}0.72 & \cellcolor{green!15}0.72 & \cellcolor{green!20}0.57 & \cellcolor{green!25}0.21 & \cellcolor{green!25}0.23 & \cellcolor{green!30}0.13 & \cellcolor{green!30}0.10 & \cellcolor{green!30}\textbf{0.15} & \cellcolor{green!30}\textbf{0.15} & \cellcolor{green!30}\textbf{0.10} & \cellcolor{green!30}\textbf{0.08} \\
 & R-A & \cellcolor{green!25}0.30 & \cellcolor{green!25}0.38 & \cellcolor{green!25}0.38 & \cellcolor{green!25}0.31 & \cellcolor{green!20}0.48 & \cellcolor{green!20}0.47 & \cellcolor{green!25}0.36 & \cellcolor{green!25}0.29 & \cellcolor{green!25}\textbf{0.22} & \cellcolor{green!25}\textbf{0.25} & \cellcolor{green!30}\textbf{0.14} & \cellcolor{green!30}\textbf{0.12} & \cellcolor{green!25}\textbf{0.22} & \cellcolor{green!25}0.31 & \cellcolor{green!30}0.19 & \cellcolor{green!30}0.15 \\
\cmidrule{1-18}
\multirow{6}{*}{Adult} & G & \cellcolor{orange!10}1.07 & \cellcolor{green!15}0.68 & \cellcolor{green!20}0.41 & \cellcolor{green!15}0.79 & \cellcolor{green!30}\textbf{0.05} & \cellcolor{green!30}\textbf{0.04} & \cellcolor{green!30}0.04 & \cellcolor{green!30}\textbf{0.19} & \cellcolor{green!20}0.43 & \cellcolor{green!25}0.26 & \cellcolor{green!20}0.53 & \cellcolor{green!15}0.78 & \cellcolor{green!30}0.06 & \cellcolor{green!25}0.28 & \cellcolor{green!30}\textbf{0.02} & \cellcolor{green!25}0.31 \\
 & R & \cellcolor{green!10}0.85 & \cellcolor{green!15}0.80 & \cellcolor{green!15}0.77 & \cellcolor{orange!20}1.11 & \cellcolor{green!25}0.36 & \cellcolor{green!20}0.41 & \cellcolor{green!25}0.23 & \cellcolor{green!10}0.93 & \cellcolor{green!15}0.73 & \cellcolor{green!15}0.65 & \cellcolor{green!15}0.80 & \cellcolor{green!10}0.96 & \cellcolor{green!30}\textbf{0.18} & \cellcolor{green!25}\textbf{0.36} & \cellcolor{green!30}\textbf{0.09} & \cellcolor{green!15}\textbf{0.74} \\
 & A & \cellcolor{green!25}0.26 & \cellcolor{green!30}0.18 & \cellcolor{green!25}0.35 & \cellcolor{green!20}0.56 & \cellcolor{orange!20}1.15 & \cellcolor{orange!20}1.12 & \cellcolor{green!20}0.54 & \cellcolor{green!20}0.43 & \cellcolor{orange!10}1.08 & \cellcolor{green!10}0.96 & \cellcolor{green!10}0.86 & \cellcolor{green!15}0.78 & \cellcolor{green!30}\textbf{0.00} & \cellcolor{green!30}\textbf{0.01} & \cellcolor{green!30}\textbf{0.03} & \cellcolor{green!30}\textbf{0.08} \\
 & G-R & \cellcolor{green!15}0.61 & \cellcolor{green!10}0.92 & \cellcolor{green!20}0.47 & \cellcolor{green!15}0.64 & \cellcolor{green!15}0.70 & \cellcolor{green!15}0.62 & \cellcolor{green!25}0.32 & \cellcolor{green!15}0.64 & \cellcolor{green!15}0.69 & \cellcolor{green!15}0.63 & \cellcolor{green!15}0.71 & \cellcolor{green!10}0.93 & \cellcolor{green!30}\textbf{0.09} & \cellcolor{green!20}\textbf{0.42} & \cellcolor{green!30}\textbf{0.04} & \cellcolor{green!20}\textbf{0.44} \\
 & G-A & \cellcolor{green!25}0.35 & \cellcolor{green!25}0.29 & \cellcolor{green!20}0.49 & \cellcolor{green!15}0.61 & \cellcolor{orange!10}1.05 & \cellcolor{green!10}0.97 & \cellcolor{green!10}0.84 & \cellcolor{green!10}0.95 & \cellcolor{orange!10}1.03 & \cellcolor{orange!10}1.01 & \cellcolor{green!10}0.97 & \cellcolor{green!15}0.67 & \cellcolor{green!30}\textbf{0.05} & \cellcolor{green!30}\textbf{0.06} & \cellcolor{green!30}\textbf{0.16} & \cellcolor{green!25}\textbf{0.21} \\
 & R-A & \cellcolor{green!25}0.40 & \cellcolor{green!25}0.34 & \cellcolor{green!20}0.47 & \cellcolor{green!20}0.47 & \cellcolor{orange!20}1.14 & \cellcolor{orange!20}1.11 & \cellcolor{green!10}0.80 & \cellcolor{red!30}1.42 & \cellcolor{green!10}0.87 & \cellcolor{green!10}0.86 & \cellcolor{green!10}0.90 & \cellcolor{green!20}0.52 & \cellcolor{green!30}\textbf{0.03} & \cellcolor{green!30}\textbf{0.08} & \cellcolor{green!30}\textbf{0.12} & \cellcolor{green!25}\textbf{0.28} \\
\cmidrule{1-18}
\multirow{3}{*}{German} & G & \cellcolor{orange!10}1.07 & \cellcolor{green!10}0.94 & \cellcolor{red!30}5.79 & \cellcolor{red!30}1.43 & \cellcolor{orange!20}1.11 & \cellcolor{orange!10}1.04 & \cellcolor{green!30}0.13 & \cellcolor{green!15}0.67 & \cellcolor{orange!30}1.20 & \cellcolor{orange!10}1.01 & \cellcolor{green!30}0.05 & \cellcolor{green!15}0.63 & \cellcolor{green!30}\textbf{0.00} & \cellcolor{green!30}\textbf{0.00} & \cellcolor{green!30}\textbf{0.00} & \cellcolor{green!30}\textbf{0.00} \\
 & A & \cellcolor{orange!20}1.20 & \cellcolor{orange!20}1.13 & \cellcolor{green!30}0.02 & \cellcolor{green!15}0.70 & \cellcolor{orange!10}1.08 & \cellcolor{orange!20}1.15 & \cellcolor{red!30}1.40 & \cellcolor{green!15}0.63 & \cellcolor{orange!20}1.16 & \cellcolor{orange!10}1.04 & \cellcolor{green!30}0.02 & \cellcolor{green!15}0.68 & \cellcolor{green!30}\textbf{0.00} & \cellcolor{green!30}\textbf{0.00} & \cellcolor{green!30}\textbf{0.00} & \cellcolor{green!30}\textbf{0.00} \\
 & G-A & \cellcolor{orange!10}1.02 & \cellcolor{orange!10}1.03 & \cellcolor{red!30}2.79 & \cellcolor{orange!30}1.30 & \cellcolor{green!10}0.88 & \cellcolor{green!10}0.85 & \cellcolor{green!25}0.23 & \cellcolor{green!10}0.83 & \cellcolor{green!10}0.86 & \cellcolor{green!10}0.83 & \cellcolor{green!30}0.02 & \cellcolor{green!15}0.64 & \cellcolor{green!30}\textbf{0.00} & \cellcolor{green!30}\textbf{0.00} & \cellcolor{green!30}\textbf{0.00} & \cellcolor{green!30}\textbf{0.00} \\
\cmidrule{1-18}
\multirow{1}{*}{Bank} & A & \cellcolor{green!15}0.62 & \cellcolor{green!20}0.58 & \cellcolor{green!30}0.10 & \cellcolor{green!15}0.71 & \cellcolor{orange!30}1.27 & \cellcolor{red!30}1.52 & \cellcolor{green!10}0.96 & \cellcolor{green!15}0.79 & \cellcolor{green!15}0.74 & \cellcolor{green!15}0.80 & \cellcolor{green!25}0.23 & \cellcolor{green!10}0.82 & \cellcolor{green!30}\textbf{0.04} & \cellcolor{green!30}\textbf{0.16} & \cellcolor{green!30}\textbf{0.03} & \cellcolor{green!25}\textbf{0.22} \\
\midrule
\multicolumn{2}{c|}{Average} & \cellcolor{green!15}0.61 & \cellcolor{green!20}0.60 & \cellcolor{green!10}0.89 & \cellcolor{green!15}0.64 & \cellcolor{green!15}0.75 & \cellcolor{green!15}0.76 & \cellcolor{green!20}0.54 & \cellcolor{green!15}0.63 & \cellcolor{green!15}0.63 & \cellcolor{green!20}0.58 & \cellcolor{green!25}0.37 & \cellcolor{green!20}0.50 & \cellcolor{green!30}\textbf{0.06} & \cellcolor{green!30}\textbf{0.12} & \cellcolor{green!30}\textbf{0.05} & \cellcolor{green!30}\textbf{0.16} \\
\bottomrule
\end{tabular}
\vspace{-1em}
\end{table*}

\subsection{RQ3: How does \name perform under various testing frameworks?}

To answer this research question, we employ four state-of-the-art NN fairness testing frameworks to evaluate \name and all baselines. 
Specifically, we follow the original configurations of these frameworks and used them to generate IDIs for both the original and the repaired model. 
We denote the number of IDIs detected on the original model as $\text{N}_{\text{ori}}$ and those detected on the repaired model as $\text{N}_{\text{repair}}$.
To measure repair effectiveness under each testing framework, we report the ratio $\text{N}_{\text{repair}} / \text{N}_{\text{ori}}$, where a smaller ratio indicates that fewer IDIs are detected and thus better fairness improvement. 

The results are presented in Table~\ref{table:testing}. 
Overall, all methods reduce model bias, broadly consistent with our earlier findings.
However, we find that these testing frameworks reveal model defects more precisely and effectively than evaluation based on uniform sampling, especially in certain settings. 
For instance, all baseline methods especially FLIP and CARE perform poorly on the German dataset, and in some cases the repaired models even exhibit more IDIs than the original ones (ratios greater than 1). 
In contrast, our method shows consistently better performance, effectively reducing bias across all settings. 
In particular, \name achieves average ratios of 0.06, 0.12, 0.05, and 0.16 under the four testing tools, corresponding to 94\%, 88\%, 95\%, and 84\% reductions in IDIs, respectively.
These results are in close agreement with our earlier sampling-based generalization evaluation  (i.e., nearly 90\% reduction).
By comparison, the best-performing baseline, GRFT, reduces IDIs by no more than 70\%, further confirming the superior generalization of our method against all baselines.

\subsection{RQ4: How do the two key steps of \name impact repair?}
\label{sec:exp-rq4}
To answer this question, we first examine the impact of the first step, which performs a progressive bounds tightening process.
Specifically, we track the dynamics of the two losses minimized during this step: $\mathcal{L}_\text{bce}$, the BCE loss computed on a small subset of training data $\mathcal{D}{\text{c}}$, and $\mathcal{L}_\text{fair}$, a fairness-aware loss measuring the relative $\ell_1$-norm reduction between the current and original bounds.
Due to space limitations, we present only a subset of the results in Figure~\ref{fig:exp-rq4-s1}; the full version is shown in Appendix.
We observe that $\mathcal{L}_{\text{fair}}$ steadily decreases as epochs progress, eventually reaching a relatively low level. 
This indicates that the concrete bounds after repair are significantly tightened compared to before, demonstrating Step 1 successfully calibrates and reduces model bias. 
Meanwhile, $\mathcal{L}_{\text{bce}}$ remains stable, showing that model's knowledge is largely preserved. 
These results highlight that this step effectively balances fairness enhancement and accuracy retention, with $\mathcal{L}_{\text{fair}}$ playing a pivotal role in calibrating the model without sacrificing its fidelity.

\begin{figure}[!t]
\centering
\includegraphics[width=1.0\linewidth]{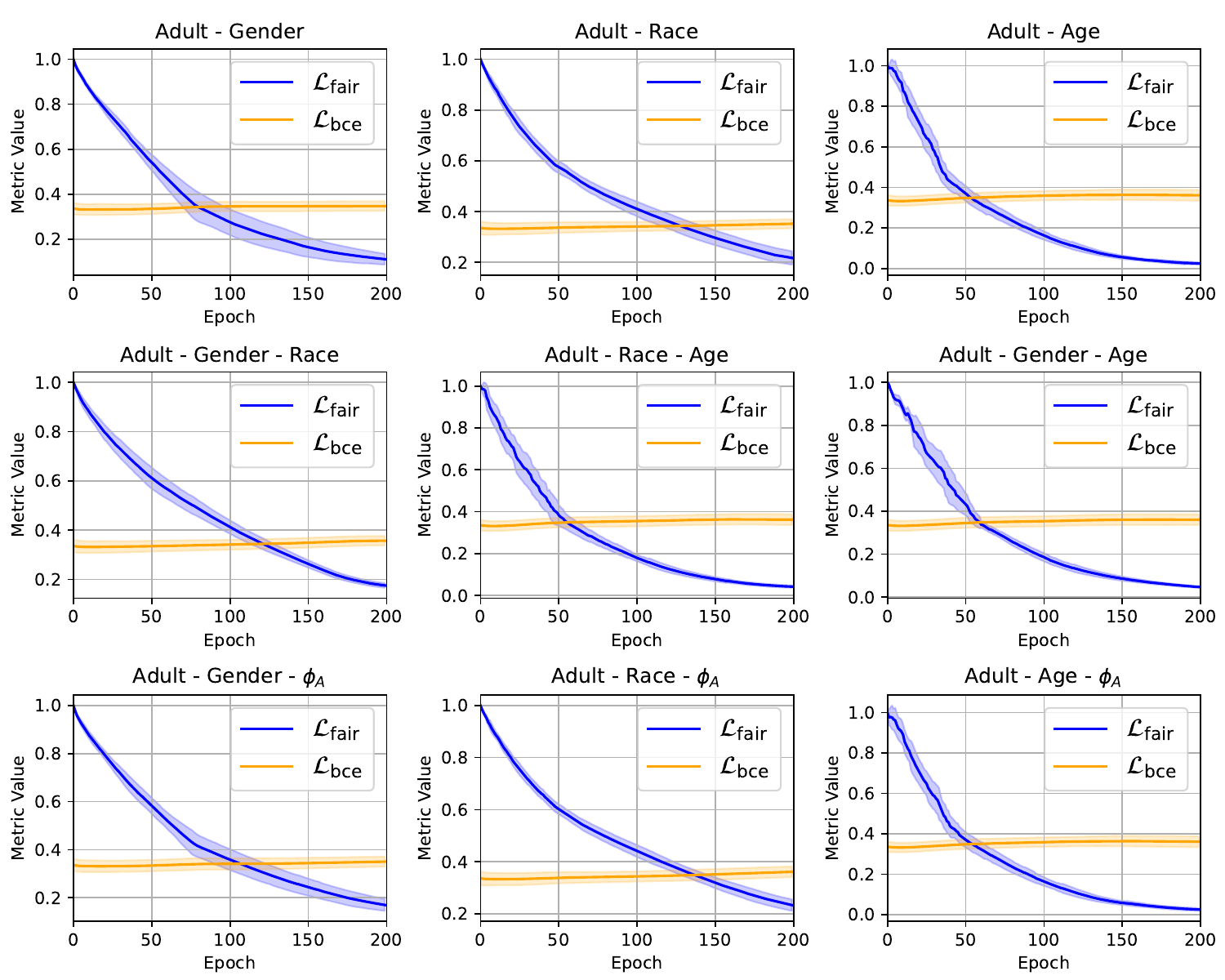}
\caption{Loss dynamics during the bounds tightening in Step 1 of \name. Shaded areas denote the variance across 10 runs.}
\label{fig:exp-rq4-s1}
\vspace{-0.5em}
\end{figure}

We further analyze the contribution of the second step in \name.
Figure~\ref{fig:exp-rq4-s3} compares \name with the naive repair method that relies on loose concrete bounds.
As shown in the top panel, models repaired by \name consistently achieve higher accuracy, suggesting that the naive method leads to excessive modifications of the final layer as it requires fairness constraints are satisfied across the box formed by concrete bounds.
This observation is further supported by the bottom panel, where we report the optimal values of the respective MILP formulations solved by Gurobi.
\name consistently yields smaller solutions, indicating that the tighter symbolic bounds enable less adjustment to the model.

\begin{figure}[t]
\centering
\includegraphics[width=1.0\linewidth]{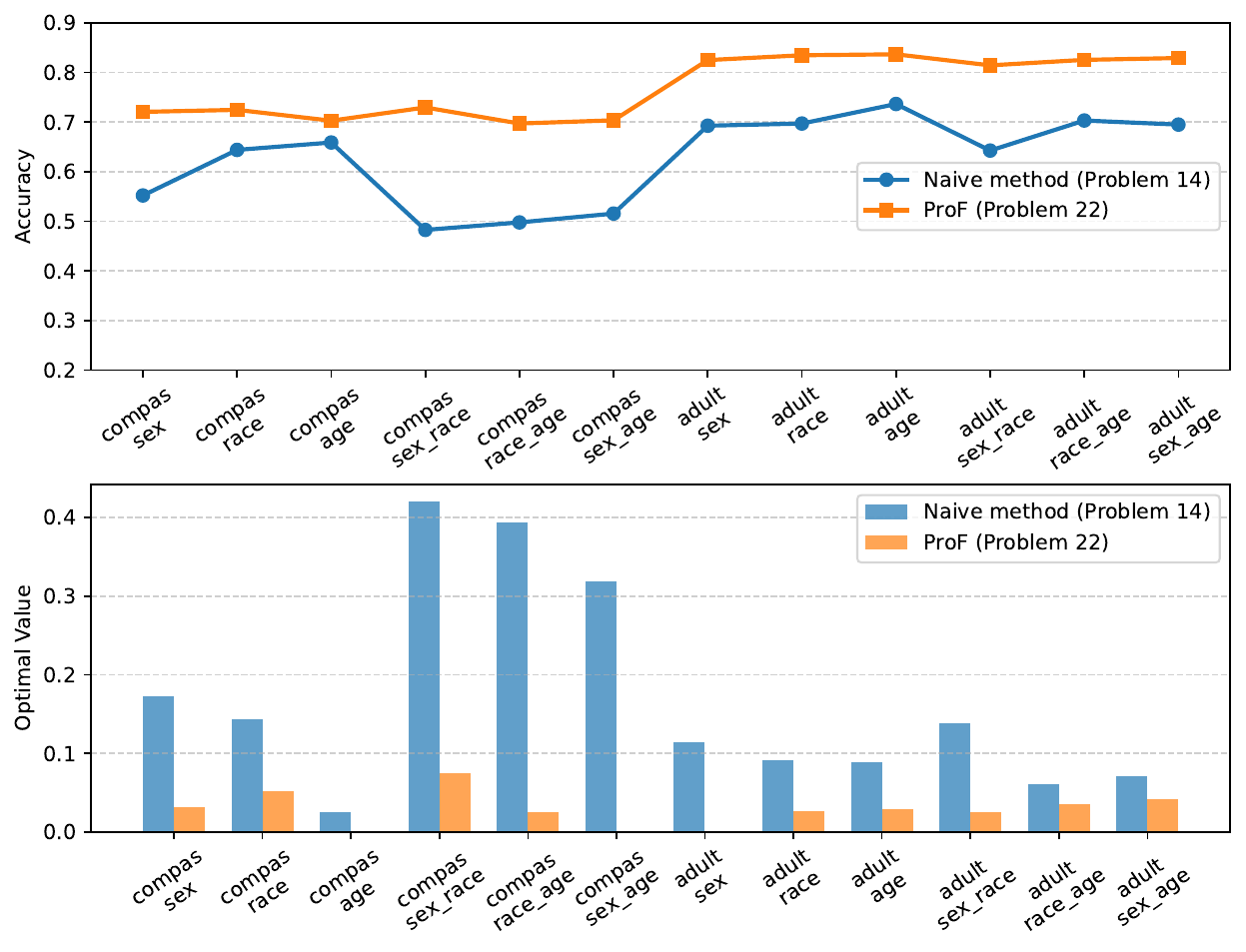}
\caption{Comparison between the naive repair method and \name. Top: Accuracy of models repaired by the naive method and \name. 
Bottom: Optimal values of the MILP formulated by the naive method (Problem~\ref{eq:concrete-repair}) and \name (Problem~\ref{eq:milp}). }
\label{fig:exp-rq4-s3}
\vspace{-0.9em}
\end{figure}

\section{Discussion}

\textbf{Scalability to more complex models:}
\name can scale to more complex models since it only requires the final layer to be linear, which holds for most DNNs. 
A potential limitation when scaling to these models is computational cost, since MILP is theoretically NP-hard and its complexity grows exponentially with the number of integer variables.
However, this number is irrelevant to the model size in our formulation, which could mitigate the exponential explosion issue. 
The remaining variables and constraints form a standard linear program solvable in polynomial time.
Nevertheless, we agree that the overall MILP complexity cannot be guaranteed to be polynomial due to its inherent NP-hardness. 

\textbf{Relation to robustness repair:} 
\name achieves provable fairness repair through two carefully designed steps. 
The first step mitigates bias by tightening concrete bounds to promote feature consistency among similar individuals, regardless of the ground truth label. 
This method is tailored for fairness and cannot apply to robustness repair directly, which requires all perturbed inputs to be classified correctly. 
In the second step, we formulate a unified constraint solving problem. 
The key challenge here is the presence of non-linear terms in the MILP, which we address by introducing the dual theorem. 
This technique can also be adapted to address a similar problem in robustness repair.
Beyond this issue, fairness repair exhibits another distinctive challenge: it involves disjunctive constraints that cannot be handled by existing solvers. 
To address this, we carefully design and introduce integer variables to transform disjunctive constraints into linear form.
Overall, while some techniques in the second step can be adapted to robustness repair, the first step is fairness-specific, which makes \name currently unique to fairness. 

\section{Related Work}
\textbf{NN Fairness.}
Fairness for DNNs is commonly categorized into \textit{group fairness} and \textit{individual fairness}. Group fairness considers whether different demographic groups receive statistically similar outcomes~\cite{feldman2015certifying, hardt2016equality}. 
Despite its simplicity in metric design, it only captures statistical parity and may result in unfairness at the individual level. 
In contrast, individual fairness~\cite{zemel2013learning, dwork2012fairness} stipulates that similar individuals should be treated similarly by the model, typically formalized as the constraint that predictions remain consistent when only the sensitive attributes are changed. 
It enables finer, instance-level analysis and broader use of testing and verification techniques.

Various methods have been developed to test and verify individual fairness in DNNs.
ADF~\cite{zhang2020white} and EIDIG~\cite{zhang2021efficient} search for IDIs near the decision boundary by leveraging the loss function gradient. 
NeuronFair~\cite{zheng2022neuronfair} further optimizes test generation by computing gradients only for biased neurons identified via sensitivity analysis. 
More recently,~\cite{monjezi2025fairness} introduces extreme value theory to fairness testing.
It models the worst case counterfactual bias and develops a randomized test-case generation algorithm to collect tail samples with statistical guarantees.
Beyond testing, several works have been proposed to formally verify NN fairness.
For example,~\cite{sun2021probabilistic} learns Markov Chains from a given model to formally verify group fairness with probabilistic guarantees.
On the other hand, DeepGemini~\cite{xie2023deepgemini} encodes individual fairness constraints into SMT formulas but is limited in scalability. 
Fairify~\cite{biswas2023fairify} improves verification efficiency by decomposing the problem and pruning neurons. 
FairQuant~\cite{kim2024fairquant} further enhances scalability and precision via abstraction-refinement, and additionally quantifies the proportion of inputs that are certifiably fair or unfair.

Regarding unfairness mitigation, several works~\cite{chen2024isolation, quan2025dissecting, sun2022causality, li2023fairer} use heuristic algorithms to correct model bias but lack formal guarantees. 
Shifty~\cite{giguere2022fairness} presents a fairness training method, offering high-confidence fairness guarantees under demographic shift.
NeuFair~\cite{dasu2024neufair} and AutoRIC~\cite{sun2025autoric} are
recently proposed methods that target group fairness repair: the former leverages simulated annealing to optimize repair solutions with statistical guarantees, while the latter uses constraint solving. 
Unlike these methods focusing on group
fairness, this work aims to provide \textit{deterministic}, \textit{provable} guarantees for \textit{individual fairness repair}.

\textbf{NN Repair.}
There exists a broader line of work on general NN repair that targets other correctness properties.
Among them, non-provable methods~\cite{sun2022causality, usman2021nn, henriksen2022repairing, sohn2023arachne, chen2024interpretability, ma2024vere} typically follow a two-step pipeline: first localizing faulty neurons or parameters, then applying heuristic techniques to adjust them. 
These methods often rely heavily on sufficient data and lack rigorous guarantees.
To provide guarantees, several recent works~\cite{sotoudeh2021provable, reassure, tao2023architecture} leverage constraint solvers to calculate parameter changes ensuring property satisfaction.
However, these methods are not specifically designed for fairness and thus cannot be directly applied to fairness repair.

\section{Conclusion}
We present \name, a provable fairness repair framework for DNNs. 
It leverages interval bound propagation to calculate concrete bounds that soundly capture the model behavior in feature space, and iteratively tightens these bounds to calibrate the biased model.
\name further synthesizes symbolic bounds to formulate a precise constraint-solving problem and introduces the duality theorem to eliminate non-linear operations to construct an MILP that provides provable guarantees for repair.
Extensive experiments demonstrate that \name significantly outperforms the state-of-the-art in terms of both effectiveness and generalization.
Moreover, \name can handle multiple sensitive attributes and relaxed fairness definitions.

\section*{Acknowledgments}
This work is supported by the NSFC (No. U21B2001, No. 62402150) and the Key R\&D Program of Zhejiang
under Grant (No. 2025C01083).

\clearpage
\bibliographystyle{IEEEtran}
\bibliography{IEEEabrv, wx}

\clearpage
\appendices
\begin{figure*}[!b]
\centering
\includegraphics[width=1.0\textwidth]{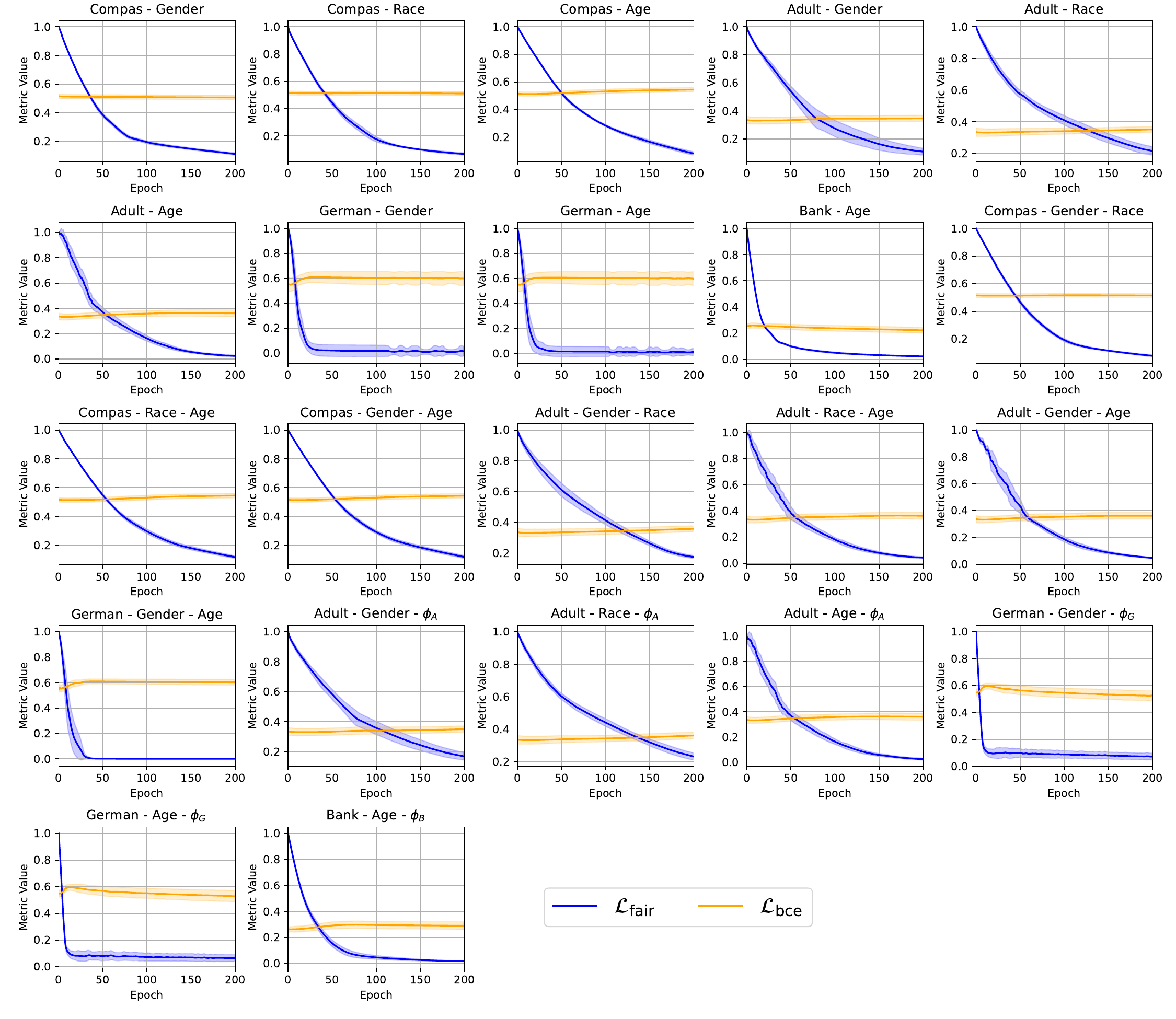}
\caption{Full results of loss dynamics during the progressive bounds tightening process in Step 1 of \name, with shaded areas representing the variance across multiple runs.}
\label{fig:exp-rq3-s1-full}
\end{figure*}




\end{document}